\documentclass[12pt]{article}

\usepackage[export]{adjustbox}
\usepackage[margin=1in]{geometry}

\usepackage{xfrac}
\usepackage{times}
\usepackage{graphicx}
\usepackage{amsmath,amsthm,amssymb,amsfonts}
\usepackage[colorlinks,linkcolor=black,bookmarksopen,
  bookmarksnumbered,citecolor=black,urlcolor=black]{hyperref}
\usepackage[]{algorithm2e}
\usepackage{mathtools}
\DeclarePairedDelimiter\ceil{\lceil}{\rceil}
\DeclarePairedDelimiter\floor{\lfloor}{\rfloor}

\newcommand{\R}{\mathbb{R}}

\newcommand{\Z}{\mathbb{Z}}

\usepackage{enumitem}
\usepackage{pifont}
\usepackage{tikz-cd}
\usepackage{tikz}
\usetikzlibrary{arrows,automata} 
\usetikzlibrary{positioning}
\usetikzlibrary{calc}

\usepackage{color}

\graphicspath{{Figures/}}

\theoremstyle{plain}
\newtheorem{theorem}{Theorem}[section]
\newtheorem{lemma}[theorem]{Lemma}

\theoremstyle{definition}
\newtheorem{definition}[theorem]{Definition}

\newtheorem{example}[theorem]{Example}

\newcommand{\MCMF}{\textsc{Mult-Com Max-Flow }}
\newcommand{\RMF}{\textsc{Ratio Max-Flow }}

\newcommand{\IRMF}{\textsc{Int Ratio Max-Flow }}
\newcommand{\MaxFlow}{\textsc{Max-Flow }}
\newcommand{\MinCost}{\textsc{Min-Cost Flow }}

\newcommand{\MCMFs}{\textsc{Mult-Com Max-Flow}}
\newcommand{\RMFs}{\textsc{Ratio Max-Flow}}

\newcommand{\IRMFs}{\textsc{Int Ratio Max-Flow}}

\newcommand{\Mnd}{\mathcal{M}} 
\newcommand{\mytilde}{\raise.17ex\hbox{$\scriptstyle\mathtt{\sim}$}} 

\usepackage{subcaption}

\title{A Tight Max-Flow Min-Cut Duality Theorem for NonLinear Multicommodity Flows}

\author{
  \hspace*{-0.5in} Matthew Broussard\footnote{\href{mailto:matthewbrouss@gmail.com}{matthewbrouss@gmail.com}} \hspace*{1.6in} 
  Bala Krishnamoorthy\footnote{\href{mailto:kbala@wsu.edu}{kbala@wsu.edu}}\\
  TD Bank, N.A., USA 
  \hspace*{1in}  
  Washington State University, USA
  }

\begin{document}

\maketitle

\begin{abstract} \label{Abstract}
The Max-Flow Min-Cut theorem is the classical duality result for the  \MaxFlow problem, which considers flow of a single commodity.
We study a multiple commodity generalization of \MaxFlow in which flows are composed of real-valued $k$-vectors through networks with arc capacities formed by regions in $\mathbb{R}^k$.
Given the absence of a clear notion of ordering in the multicommodity case, we define the generalized max flow as the feasible region of all flow values.

We define a collection of concepts and operations on flows and cuts in the multicommodity setting.
We study the \emph{mutual capacity} of a set of cuts, defined as the set of flows that can pass through all cuts in the set.
We present a method to calculate the mutual capacity of pairs of cuts, and then generalize the same to a method of calculation for arbitrary sets of cuts.
We show that the mutual capacity is exactly the set of feasible flows in the network, and hence is equal to the max flow.
Furthermore, we present a simple class of the multicommodity max flow problem where computations using this tight duality result could run significantly faster than default brute force computations.

We also study more tractable special cases of the multicommodity max flow problem where the objective is to transport a maximum real or integer multiple of a given vector through the network.
We devise an augmenting cycle search algorithm that reduces the optimization problem to one with $m$ constraints in at most $\R^{(m-n+1)k}$ space from one that requires $mn$ constraints in $\R^{mk}$ space for a network with $n$ nodes and $m$ edges.
We present efficient algorithms that compute $\epsilon$-approximations to both the ratio and the integer ratio maximum flow problems.

\end{abstract}


\clearpage
\section{Introduction} \label{sec:Introduction}
Network flows are widely used to model transportation systems of commodities or people subject to capacity constraints.
Ford and Fulkerson \cite{FoFu1956} introduced one of the first motivating examples for network flows in 1956: a system of rail lines connecting cities.
Each city is represented by a node in the network and the rail lines connecting them form the arcs.
Each arc is then labeled with the maximum number of people its associated rail line can transport in, say, a day---this is the capacity of the arc.
Another canonical example is a network of pipes transporting water.
Pipes form the arcs of the network and junctions form the nodes.
The circumference of a pipe determines the flow rate it permits, which is specified as the capacity of the arc associated with the pipe.

In such examples, we often wish to find the maximum amount of a commodity we can transport.
We might wish to find out how many people could possibly travel from New York to Chicago by train, or the maximum flow rate we can achieve through a network of pipes leading to a sink.
This class of problems is referred to as the \MaxFlow problem.
A concept central to the study of \MaxFlow is that of \emph{cuts}, in particular the collections of arcs that, when removed, separate the starting point from our destination.
Also called \emph{disconnecting} sets (of arcs), the concept of cuts is dual to flows.
The \emph{value} $v(D)$ of a disconnecting set $D$ is the sum of the capacities of all (forward) arcs in $D$.
Ford and  Fulkerson \cite{FoFu1956} presented the classical theorem linking cuts to flow values, as well as a method for finding the maximum flow by finding a minimal cut. 
\begin{theorem}[Max-Flow Min-Cut Theorem, Ford and Fulkerson \cite{FoFu1956}] \label{thm:1comMFMC}
  The maximal flow value obtainable in a network $N$ is the minimum of $v(D)$ taken over all disconnecting sets $D$.
\end{theorem}
Noting that a disconnecting set with a minimum possible value is automatically a cut, it follows that the maximum flow value is exactly the minimum cut value.
The authors described an algorithm to identify the minimum cut in polynomial time by finding augmenting paths.

One generalization of \MaxFlow is the \MinCost problem \cite{AhMaOr1993} where each arc has both a capacity and a cost for each unit of the commodity transported.
A subset of the nodes are source (supply) nodes with given amounts of the commodity and another subset are sink (demand) nodes with specific required amounts.
The problem is to route the commodity from source nodes to sink nodes such that demands at the sink nodes are met, capacities of arcs are honored, and the total cost is minimized.
A further generalization of \MinCost to handle multiple types of commodities termed multicommodity minimum cost flows has been well studied \cite[Chap.~17]{AhMaOr1993}.

However, the direct generalization of \MaxFlow to the problem with multiple types of commodities without the consideration of costs, supply, and demand has not received as much attention.
While this generalizes the standard max-flow problem \cite{Kr20??final},
the single commodity algorithms may not generalize easily.
It is not clear how to generalize the classical result equating maximum flows with minimum cuts (Theorem \ref{thm:1comMFMC}) to the multicommodity case.
In fact, specifying the boundary of feasible regions for flows of two or more commodities is already challenging.
We focus on this problem (termed \MCMFs) with the goal of finding the set of feasible flows through the a network.

A more tractable version of this problem could seek a maximum multiple of a certain ratio of commodities that could be transported through a network.
This problem is termed \RMFs, and we could also consider a special case termed \IRMF which allows only integer multiples.
As these problems search for a single value, we may be able to modify some of the single-commodity algorithms to search for the maximum (integer) multiple.
Both \RMF and \IRMF could model practical situations.
Suppose a factory manufacturing cars uses tires transported from one depot and chassis transported from another.
In this case, the desired ratio would be $4:1$ tires to chassis, and combinations which are not integer multiples of this ratio do not form a full car and are thus not usable.
This is a situation that \IRMF models.
On the other hand, we could model a dye dispensary which mixes a dark green dye from one part yellow and two parts blue using \RMFs, since we can make partial units of the green dye.

\vspace*{-0.05in}
\subsection{Our Contributions} \label{ssec:Contributions}
\vspace*{-0.03in}

We investigate a general case of \MCMF with $k$ commodities where the capacity restrictions on arcs are defined as regions of $\R^k$.
With a goal toward developing a tight multicommodity max-flow min-cut theorem, we define several new concepts related to flows and cuts in this setting.
As it is not clear how one would define a \emph{maximum} flow in the multicommodity setting, we analyze the set of all feasible flows instead.
We study \emph{pseudoflows} that could violate capacity constraints while honoring flow balance at nodes, and \emph{local flows} over a cut that assign feasible flows for all arcs in the cut (see Section \ref{ssec:defs_flows}).
We also define the concept of \emph{mutual capacity} as the correct generalization of cut capacity in the single commodity case to our general setting (Section \ref{sec:mutual_capacity}).
Our main theorem presents a tight duality result in the multicommodity setting:
{
\renewcommand{\thetheorem}{\ref{thm:MCMFMC}}
  \begin{theorem} 
  The set of feasible flows through a network is precisely the mutual capacity of all cuts in the network.
\end{theorem}
}
\noindent Furthermore, our proof is a constructive algorithm, so calculation of the max flow, which generalizes to the feasible region of flows, gives us the set of all flows that have a particular flow value in the feasible region.

While we address a less general version of the max flow problem than Krishnan \cite{Kr20??final} as we restrict ourselves to subsets of $\mathbb{R}$ for our fields of coefficients, we use the restriction to give more detailed descriptions of the behavior of network flows.
We present a method for calculating the max flow for any network with capacities in $\mathbb{R}^k$. 
Additionally, we prove our results from first principles rather than using the machinery of category theory as done by Krishnan \cite{Kr20??final}.

We use our duality result to construct an algorithm to solve the general \MCMF problem.
In addition, we present a class of instances of the problem on which our algorithm calculates the max flow region exponentially faster than the default brute force approach.

Motivated by the work of Leighton and Rao \cite{LeRa1999}, we also consider more tractable cases of the \MCMF problem that aim to send a maximum multiple of a particular ratio of commodities.
We present efficient algorithms for two problems that take advantage of this optimization step.
We show that when the capacities of the network are bounded and reducible, these algorithms return accurate results and also terminate in finite time.
Furthermore, these algorithms reduce both the size 
and the dimension of the optimization approach by using the structure inherent in the network.
We improve the complexity of the optimization problem from requiring $mn$ constraints in $\mathbb{R}^{mk}$ to $m$ constraints in at most $\mathbb{R}^{(m-n+1)k}$ when there are $n$ nodes and $m$ arcs.

\subsection{Previous Work} \label{ssec:Previous}

While multicommodity \emph{minimum cost} flow problems are well studied in the context of linear and integer programming \cite[Chap.~17]{AhMaOr1993}, multicommodity maximum flow problems have not received as much attention.
Leighton and Rao \cite{LeRa1999} studied a version of the multicommodity max flow problem  where a network with scalar capacities on each arc is given for multiple commodities, each with its own source and sink, as well as a desired ratio of commodities.
(In contrast, motivated by Krishnan's work \cite{Kr20??final}, we consider more general classes of capacities.)
The max flow problem here was to maximize the ratio transported from source to sink subject to the restriction that the sum of all commodities through an arc cannot exceed its scalar capacity.
They also defined an analogue to the min cut called the sparsest cut, which is a cut that minimizes the ratio of the capacity of the cut to the sum of the demands for which the source and sink are separated by the cut.

They summarized several results for this problem, including ones from previous work, showing that the max flow and sparsest cut are equal for two-commodity flows, that the max flow and sparsest cut are equal provided that the graph constructed with edges $(s_i,t_i)$ for each commodity $i$ has no set of three disjoint edges and no triangle along with a disjoint edge, and further information on special cases.
Their main result was a general proof that the max flow is within a $\Theta (\log_2 n)$-factor of the sparsest cut.
They also developed an algorithm for determining an approximate max flow.

Ghrist and Krishnan \cite{GhKr2013} and then Krishnan \cite{Kr20??final} used sheaf theory to develop a notion of directed homology and cohomology.
This framework allowed the characterization of flows and cuts as directed homology and cohomology classes.
The ``min cut'' was defined as the intersection over all cuts of the Minkowski sum of arc capacities in cuts, and the max flow was replaced by the feasible region of flows.
Poincar\'e duality was then used to establish a general max-flow min-cut theorem.
However, the general theorem presented by Ghrist and Krishnan \cite{GhKr2013} had duality gaps in that the region given as the set of feasible flows could include some flow values that cannot be achieved.

The manuscript titled \emph{Flow-Cut Dualities for Sheaves on Graphs} \cite{Kr20??final} ends with two cases for which the max flow is equal to the min cut.
If all paths from source to sink pass through a minimum, i.e., have an arc which has a capacity that is a subset of each capacity along a path, or if the capacities come from a lattice-ordered semi-module, then the min cut and feasible region for flows are equal.

However, this approach does not give a general tractable solution to the multicommodity flow problem since the model given for multicommodity flow, which uses Minkowski sums and intersections, does not form a lattice order (in general, $(A +_M B) \bigcap (A +_M C) \neq A +_M (B \bigcap C)$).
Furthermore, there exist networks where some paths do not pass through a minimum, and indeed transforming the bounds on arcs so as to make all paths pass through a minimum modifies the real network capacity.
Our approach develops a tight bound on max flow using arguments from first principles.
With the bound established, we show how our results fit into Krishnan's framework by modifying certain definitions. 

Even, Itai, and Shamir \cite{EvItSh1976} studied timetable and integral multicommodity flow problems, showing both directed and undirected versions of the multicommodity flow problems are NP-complete.
In the undirected case, the flows are assigned a direction ($\pm$) even though the edges themselves are undirected.
The capacity functions considered in this work were restricted to natural numbers as upper bounds on the simple sums of flows of all commodities.
Our multicommodity flow problem allows highly general nonlinear capacity functions, and hence includes the version studied by Even, Itai, and Shamir as a special case.
As such, our multicommodity flow problem turns out to be NP-complete as well.

\section{Problems and Definitions} \label{sec:definitions}
We first formally define the three multicommodity max-flow problems of interest.
We then show that the main problem turns out to be NP-complete, as it covers as special cases a class of 2-commodity flow problems that are known to be NP-complete.
We then list various definitions related to networks and flows that we use throughout the paper.
Several of them are standard concepts defined in previous work, but are listed here for the sake of completeness.
Definitions without citations are our original contributions.

Following standard convention \cite{AhMaOr1993}, we let $n$ and $m$ denote the number of nodes and arcs in the network, and $k$ denote the number of commodities.
We denote the capacity of arc $a$ by $C_a$, and will drop the subscript when the arc is clear from context.
Furthermore $\mathcal{C}$ denotes a set of cuts, and is also referred to as a cut set.

\begin{definition}
  [Multicommodity Max Flow \cite{Kr20??final}] Given a network with capacities in $k$ commodities specified as subsets of $\R^k$, \MCMF aims to find the set of feasible flows through the network.
\end{definition}

\begin{definition}
  [Ratio Max Flow \cite{LeRa1999}] Given a network with real capacities in $k$ commodities, \RMF aims to find the flow that is the maximum multiple of a desired ratio of the $k$ commodities.
\end{definition}

\begin{definition}
  [Integer Ratio Max Flow] Given a network with real capacities in $k$ commodities, \IRMF aims to find the flow that is the maximum integer multiple of a desired ratio of the $k$ commodities.
\end{definition}

\medskip
Note that the arc capacities we consider (see Definition \ref{def:capacity}) are fairly general---in particular, we do not assume they are convex, or even compact, subsets of $\R^k$.
Even, Itai, and Shamir studied versions of the integer multicommodity flow problem \cite{EvItSh1976}, and their connections to certain timetable scheduling problems.
They showed these classes of problems are NP-complete for both directed and undirected networks.
These problems turn out to be special cases of \MCMFs.

More precisely, \textsc{D2CIF} is the directed 2-commodity integer flow problem where the directed network has a pair $(s_i,t_i)$ of origin and destination nodes for $i=1,2$, and integer \emph{requirements} $R_i \geq 0$  on the amount of outflow of commodity $i$ from $s_i$.
The two source nodes $s_1, s_2$ could be identical, and so could be the destination nodes $t_1, t_2$.
The capacity restrictions on the edges are nonnegative integers that present upper bounds on the simple sum of flows of both capacities.
The goal of \textsc{D2CIF} is to decide if a nonnegative integer feasible flow exists that satisfies the requirements.
The related \textsc{U2CIF} considers the same problem on \emph{undirected} networks, but with the flows still given a positive or negative direction.
The capacity restrictions are now applied to the sum of absolute values of directed flows of the two commodities.

Given an instance of \textsc{D2CIF} (or \textsc{U2CIF}), we can create an equivalent instance of \MCMF as follows.
We add extra nodes $s,s',t$, and edges $(s,s'), (s',s_1), (s',s_2)$, as well as  $(t_1,t), (t_2,t)$ (directed or undirected as needed).
The capacity restriction for edge $(s,s')$ is set as $(f_1,f_2) \geq (R_1,R_2)$ and $f_1, f_2 \in \Z$, the set of integers, i.e., the given \emph{requirements} are set as elementwise lower bounds on the flow of the two commodities that are also required to be integral.
All other added arcs are set as uncapacitated.
It follows that the answer to \textsc{D2CIF} is \texttt{Yes} if and only if a solution for the instance of \MCMF exists, i.e., the set of feasible flows is non-empty.

\medskip
Since \textsc{D2CIF} and \textsc{U2CIF} were shown to be NP-complete \cite{EvItSh1976}, it turns out that \MCMF is NP-complete as well.
While we focus on instances of \MCMF where the capacities are fairly general, specific instances may admit efficient solution approaches, e.g., when the capacities are all closed hyperrectangles specifying simple upper bounds on flows of individual commodities.

\subsection{Definitions on Networks} \label{ssec:defs_networks}

\begin{definition}[Single commodity Network \cite{AhMaOr1993}]
  A single commodity network $N$ is a directed graph with node set $V$ and arc set $E$ such that each arc $a \in E$ has exactly one capacity $C_a$.
  These capacities are regions in $\mathbb{R}$ ranging from  $0$  to a maximum value.
  We denote this region by noting only the maximum value.
  The network has a distinguished source node $s$ and terminus node $t$.
\end{definition}

\begin{definition}[$k$-commodity Network]
  A $k$-commodity network $N$ for $k \geq 2$ is a directed graph with node set $V$ and arc set $E$ such that each arc $a \in E$ has exactly one capacity $C_a$.
  These capacities are regions embedded in $\mathbb{R}^k$, where each dimension corresponds to a particular commodity.
  Arcs are directed as a way to fix an orientation with regards to the capacity region, i.e., flow in the direction of the arc is considered positive flow along the arc and flow in the opposite direction is considered negative.
  We can reverse the orientation of the arc if we multiply all vectors in the arc's capacity by $-1$.
  The network has a distinguished source node $s$ and terminus node $t$.
\end{definition}

Note that it is possible for flow of one commodity to be in the positive direction while another commodity flows in the negative direction.
For instance, suppose we have the (oriented) arc $(a,b)$ in the network and its capacity includes the vector $(2,-3)$.
Hence the network can transport two units of commodity 1 from $a$ to $b$ so long as it also transports three units of commodity 2 from $b$ to $a$.
We could also consider the arc in the orientation $(b,a)$, whose capacity includes the vector $(-2,3)$.
This represents the same flow in the network, i.e., two units of commodity 1 from $a$ to $b$ and three units of commodity 2 from $b$ to $a$.

\begin{definition}[Opposite Orientation]
  An arc $(a,b)$ with capacity $C \subseteq \mathbb{R}^k$ is associated with the arc $(b,a)$ with capacity $-C \subseteq \mathbb{R}^k$.
  These two arcs are said to have opposite orientations.
\end{definition}

Since our focus is on multicommodity flows, we use the terms ``network'' and ``$k$-commodity network'' interchangeably.
We say ``single commodity network'' when referring to cases where there is only one commodity.

\begin{example} \label{ex:MultiComNetworkEx}
  Consider a $2$-commodity network with $V=\{s, v_2, t\}$, $E = \{ (s,v_2), (v_2,t), (s,t)\}$, and capacities on the three arcs as shown in Figure \ref{fig:MultiComNetworkEx}.
  \begin{figure}[htp!]       
    \centering
    \includegraphics[scale=1.25]{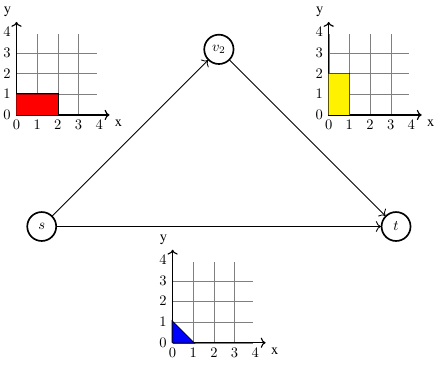}
    \caption{
        Example of a 2-commodity network.
        Capacities of the arcs are the regions shown in red, yellow, and blue.}
    \label{fig:MultiComNetworkEx}
  \end{figure}
\end{example}

Note that we have chosen the arc capacities in most of the examples as simple, yet nontrivial, sets in order to keep the examples accessible and yet insightful.
For instance, the capacities of the arcs in the network in Example \ref{ex:MultiComNetworkEx} (and in the related Example \ref{duality_gap_example}) are simplex convex polygons.
Capacities in Figures \ref{fig:MCMFLocalFlow}--\ref{fig:MCMFGlueing} are finite sets of 2D integer points (from $\Z^2_{\geq 0}$).
And in Example \ref{ex:Gap_example}, the capacity of arc $(2,3)$ is specified as a disjunction.
At the same time, the definitions and results apply to most general networks---no particular structure is assumed to be satisfied by the arc capacities unless specified otherwise.

\begin{definition}[Enhanced Network \cite{Kr20??final}]
  An enhanced network $N^E$ is a network $N$ together with an arc $e$ directed from $t$ to $s$ with the capacity constraints $x_i \geq 0$ for all commodities $i$, where $x_i$ is the flow value of commodity $i$ in the arc $e$.
\end{definition}
\noindent We assume that all networks are equipped with this arc $e$ by default, and use the term ``network'' to refer to enhanced networks.
Likewise, we omit $e$ from diagrams and examples for the brevity of illustration.
We also denote the arc set of the enhanced network by $E$.

\begin{definition}[Boundary Functions \cite{Kr20??final}]
  The functions $\partial_-, \partial_+: E \to V$ provide the tail or head of an arc, respectively.
  That is, $\partial_-(u,v)=u$ and $\partial_+(u,v)=v$ for all arcs $(u,v) \in E$.
\end{definition}
\noindent Note we often use $a$ or $e$ to denote arcs rather than listing them as explicit ordered pairs.

\begin{definition}[Capacity \cite{Kr20??final}] \label{def:capacity}
  The capacity of an arc is defined as the set of vectors in $\mathbb{R}^k$ which are permissible for that arc, i.e.,
  all possible amounts of the $k$ commodities that we can transport across that arc at the same time.
\end{definition}

Capacities can be quite general sets of vectors in $\mathbb{R}^k$.
However, we restrict ourselves to two-dimensional polytopes or integer points within the same in our example networks for ease of calculation and visualization.
Our results do hold for the most general forms of capacity unless noted otherwise.

\begin{definition}[Cut] \label{def:cut}
  A \emph{cut} in a $k$-commodity network is a partition of nodes in $V$ into two sets $S$ and $T$ such that $s \in S$ and $t \in T$.
  We associate with each cut all arcs that originate at nodes in $S$ and terminate at nodes in $T$, i.e., forward arcs, in their original orientation.
  Similarly, we associate all arcs that originate at nodes in $T$ and terminate at nodes in $S$, i.e., backward arcs, in their opposite orientation.
  The capacity of the cut is $v(\{a_i\})$ where $\{a_i\}$ is the set of appropriately oriented arcs in the cut.
\end{definition}

Note that we study only $s,t$-cuts, i.e., partitions that separate $s$ and $t$.
More generally, any partition of $V$ may be referred to as a cut in some literature \cite{AhMaOr1993}.
We further note that, unlike in the single commodity case, forward \emph{and} backward arcs can contribute to the capacity of a cut.
Also note that $e$ is a backward arc over all cuts in our setting.
 
\begin{definition}[Arc Set]
  The \emph{arc set} $E(c)$ of a cut $c$ is the collection of forward and backward arcs over the cut, i.e., arcs from $S$ to $T$ or vice-versa, along with their orientations.
  The arc set of a set of cuts $\mathcal{C}$ is similarly defined as $E(\mathcal{C}) = \cup_{c \in \mathcal{C}} E(c)$.
\end{definition}

\begin{definition}[Value Function]
  We denote the capacity of a set of arcs $\{a_i\}_{i\in I}$ by $v(\{a_i\})$ and define it as the Minkowski sum of the capacities of the individual arcs.
  We refer to $v$ as the \emph{value function}.
  When $\{a_i\}$ is the arc set of a cut and the arcs are oriented from $s$ to $t$, we call $v(\{a_i\})$ the capacity of the cut.
\end{definition}

In practice, we exclude the capacity of the arc $e$ in calculations of a network's capacity.
Recall that $e$ by definition cannot transport commodities from $s$ to $t$ as it takes only nonnegative values from $t$ to $s$.
Hence the inclusion of $e$ in the value function does not give any useful information on the network's capacity.

\begin{definition}[Total Capacity \cite{GhKr2013}]
  The \emph{total capacity} of a set of cuts is the intersection of their capacities.
\end{definition}

\begin{definition}[Mutual Capacity]
  The \emph{mutual capacity} $U_m(\mathcal{C})$ of a set of cuts $\mathcal{C}$ is the set of flow values $f$ for which there exists an assignment of flow values to the arcs in the arc sets of $\mathcal{C}$ that respect the capacity constraints and give a total value of $f$ over each cut in $\mathcal{C}$.
\end{definition}

Note that the mutual capacity must be a subset of total capacity, but the two need not be equal.
See Example \ref{duality_gap_example} for an illustration.

\begin{definition}[Reducible Capacity]
  A capacity $C$ for an arc is \emph{reducible} if for every flow value $f \in C$ and every flow value $0 \leq f' \leq f$ (element wise), $f' \in C$.
\end{definition}

Working with reducible capacities gives us the ability to speed up searches for a max flow since we can iteratively search increasing flows.
We can then terminate our search when we find a flow with no realization.

\subsection{Flows and Flow-like Objects} \label{ssec:defs_flows}

\begin{definition}[Flow \cite{Kr20??final}]
  A \emph{flow} on a network is a function $\phi: E \rightarrow \mathbb{R}^k$ subject to the following restrictions:
\begin{itemize}[itemsep=-0.01in]
    \item conservation: for $v \in V$, $\Sigma_{a\in \partial_-^{-1}(v)} \phi(a)=\Sigma_{a\in \partial_+^{-1}(v)}\phi(a)$;
    \item capacity constraints: for $a\in E$, $\phi(a) \in C_a$ and $\phi_e \geq 0$ for the edge $e=(t,s)$.
\end{itemize}
\end{definition}
\noindent Note that $\phi(v)=0$ for all nodes $v$ in an enhanced network $N^E$.


\begin{example} \label{ex:MultiComFlow}
  Figure \ref{fig:MCFlowEx} shows a flow in the 2-commodity network presented in Example \ref{ex:MultiComNetworkEx}.
  \begin{figure}[htp!]
  \centering
  \includegraphics[scale=1.25]{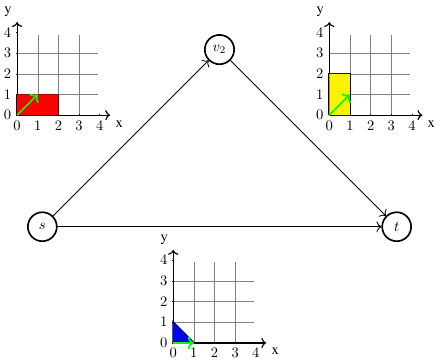}
  \caption{Example of a multicommodity flow in the network shown in Figure \ref{fig:MultiComNetworkEx}.
  The green vectors denote flow values for the two commodities ($x,y$) on each arc.}
  \label{fig:MCFlowEx}
  \end{figure}
\end{example}

\begin{definition} [Flow Value on a Network \cite{Kr20??final}]
  The \emph{flow value} of a given flow \emph{on a network} is the net value assigned to the arcs in a cut set for that network.
  Equivalently, a flow on an enhanced network has a value equal to the value assigned to arc $e$.
\end{definition}

We prove in Lemma \ref{flow_equality} that a flow takes the same total value across each cut, so the choice of the cut in the above definition is immaterial.
For example, the flow value of flow shown in Figure \ref{fig:MCFlowEx} is the net outflow of $s$, which is $(1,1)+(1,0)=(2,1)$.

\begin{definition} [Flow Value on an Arc \cite{Kr20??final}] \label{def:FlowValue}
  The \emph{flow value on an arc} is the value assigned to that arc by a flow.
\end{definition}

We define a pseudoflow by loosening the capacity constraints rather than the conservation constraints.

\begin{definition} [Pseudoflow]
  A \emph{pseudoflow} on a network is a function $\phi: E \rightarrow \mathbb{R}^k$ subject to flow conservation, i.e., for $v \in V $, $\Sigma_{a \in \partial_-^{-1}(v)} \phi(a) = \Sigma_{a\in \partial_+^{-1}(v)} \phi(a)$.
\end{definition}

A pseudoflow can be considered as an ideal transportation situation where transports have no limits.
In our approach, we will find realizations of particular flow values by transforming pseudoflows into flows.

We are also interested in assigning values to a subset of arcs.
Our first definition gives terminology to discuss a particular assignment of flow values to a particular cut.

\begin{definition}[Local Flow]
   Given a set of cuts $\mathcal{C}$, a \emph{local flow} $\phi_L : E(\mathcal{C}) \rightarrow \mathbb{R}^k$ is an assignment of values to the arcs in the cut set $E(\mathcal{C})$ subject to the following restrictions:
   \begin{itemize}
      \item capacity constraints: for $a \in E(\mathcal{C})$, $\phi_L(a) \in C_a$;
      \item flow value: for $c$ in $\mathcal{C}$, $\phi_L(e)$ is the sum $\sum_{a\in F}\phi_L(a)-\sum_{a\in B} \phi_L(a)$, where $F$ and $B$ are the sets of forward and backward arcs in $c$, respectively.
   \end{itemize}
   The value assigned to $e$ by $\phi_L$ is called the \emph{flow value} of the local flow.
\end{definition}

Note that since $e \in E(\mathcal{C})$ because $e$ is in all cuts, $e$ has a value assigned to it by any local flow.
Also, a local flow assigns values to the arcs in a set of cuts in such a way that net flow over each cut is equal.
When reading the following definitions, recall that each local flow is defined on its given set of cuts.

\begin{figure}[ht!]
    \centering
    \includegraphics[scale=1.1]{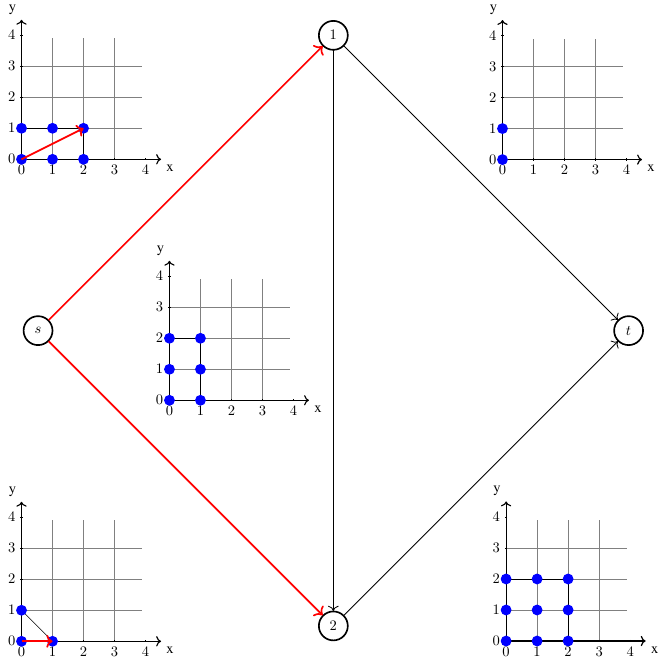}
    \caption{Example of a local flow over a cut (in red) with arc set $\{(s,1),(s,2)\}$.
    }
    \label{fig:MCMFLocalFlow}
\end{figure}

\begin{definition}[Compatible]
  A set of local flows are \emph{compatible} if for each arc $a$ in the union  of arc sets of the local flows, $\phi_L(a)$  is equal for each local flow $\phi_L$ that contains $a$ in its arc set. 
\end{definition}

Note that compatibility requires the local flows to assign the same value to $e$, so a set of local flows must each have the same flow value to be compatible.
For example, Figure \ref{fig:MCMFLocalFlow2} presents a local flow compatible with that in Figure \ref{fig:MCMFLocalFlow}.

\begin{figure}[ht!]
    \centering
    \includegraphics[scale=1.1]{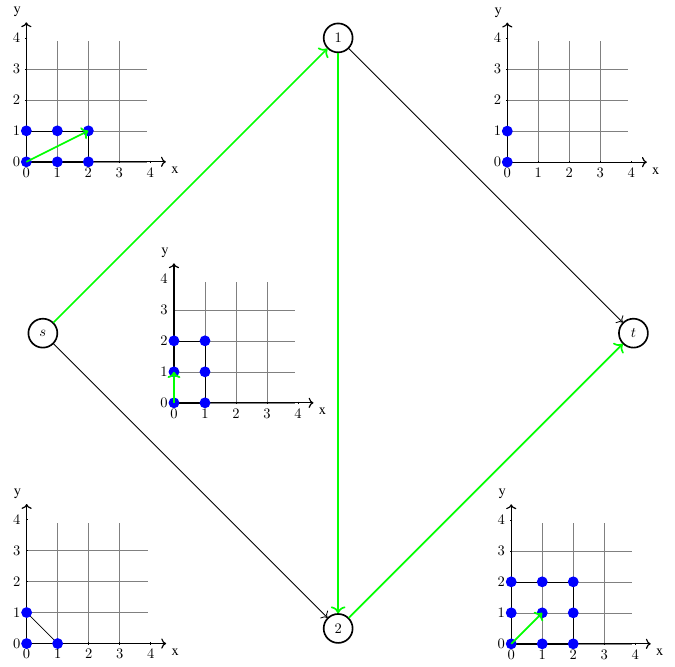}
    \caption{A second example of a local flow on the same network used in Figure \ref{fig:MCMFLocalFlow}.
    Arc set of the single cut is shown in green.
    Note that arc $(1,2)$ is a backward arc of this cut, and hence subtracts from the local flow's flow value.}
    \label{fig:MCMFLocalFlow2}
\end{figure}

\begin{definition} [Gluing] \label{def:Gluing}
  The gluing $G$ of a compatible collection of local flows $(\{\phi_L\})$ is another local flow.
  The arc set of $G$ is the union of arc sets of local flows in $\{\phi_L\}$, and likewise the cut set of $G$ is the union of the cut sets of the local flows.
  For each arc $a$ in the arc set of $G$, it holds that $G(a)=\phi_L(a)$ for any $\phi_L$ that has $a$ in its arc set.
  Since all $\phi_L$ are compatible, this is a unique definition of $G(a)$. 
\end{definition}

For instance, the local flows in Figures \ref{fig:MCMFLocalFlow} and \ref{fig:MCMFLocalFlow2} can be  glued together as they are compatible local flows.
This gluing is a local flow over a pair of cuts shown in Figure \ref{fig:MCMFGlueing}.

A gluing is local flow over a broad set of cuts built from local flows over smaller subsets of those cuts in a way that preserves the assignments given by the local flows that are glued together.
We require consistency so that no two local flows try to assign different values to the same arc.
Further, since $e$ takes a value under each local flow, we also ensure that the local flows have the same net value over every cut in the cut set.
We will build global flows as gluings of local flows that fit together in this manner.

\begin{figure}[ht!]
    \centering
    \includegraphics[scale=1.1]{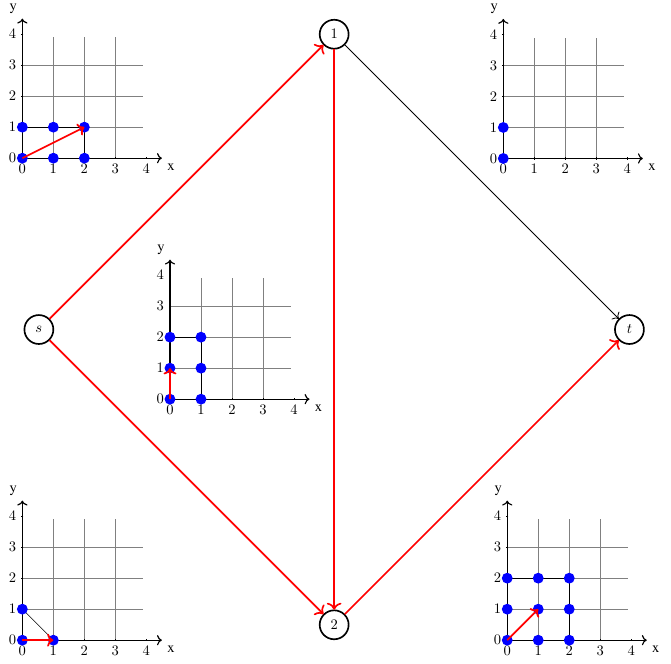}
    \caption{
        A local flow over a pair of cuts.
        This flow is obtained by gluing (see Definition \ref{def:Gluing}) together the flows in Figures \ref{fig:MCMFLocalFlow} and \ref{fig:MCMFLocalFlow2}.
    }
    \label{fig:MCMFGlueing}
\end{figure}

\subsection{Operations}  \label{ssec:defsops}

We define the operations we use to consider and compare capacities of arcs, which are defined as regions of $k$-vectors.
The basic operation is the Minkowski sum.

\begin{definition} [Minkowski Sum \cite{DeTe2014}]
  The \emph{Minkowski sum} $C$ of two sets of vectors $A$ and $B$ (denoted by $+_M$) is the set of vectors that can be written as the sum of a vector from $A$ and a vector from $B$.
\end{definition}

It makes sense to consider the capacity of a cut to be the Minkowski sum of the capacities of its arcs.
After all, we can transport $\mathbf{x}=(x_1, x_2, \dots, x_k)$ goods over a cut if and only if we can split it up into one portion for each arc in the cut which fits through that arc.
That is, we can transport $\mathbf{x}$ only if it is the element-wise vector sum of the capacities of arcs in the cut, which is the Minkowski sum.

However, when we try to compare capacities for different cuts, the Minkowski sum may give misleading results if a particular flow value requires an arc $a$ to take one value for one of the cuts it is involved with for that cut to reach a value $f$, but another cut requires a different value on $a$ to reach $f$ (see Example \ref{duality_gap_example}).

To address this discrepancy, we introduce an extension of the Minkowski sum that not only tracks the total flow value across a cut as the standard Minkowski sum does, but also records the arcs contributing to the flow and the flow values assigned to each arc.
Subsequently we define an intersection that accurately detects and discards such ``double booking''.

\begin{definition}[Enriched Minkowski Sum] \label{def:EnrichMsum}
  The \emph{enriched Minkowski sum} of the capacities of a collection $\{a_i\}$ of arcs in a cut is the Minkowski sum $+_{M}:\prod C_{a_i}\rightarrow \mathbb{R}^k$ of those capacities along with the level  set under $+_M$ in $\prod C_{a_i}$ of each vector in the Minkowski sum.
  That is, we find all vectors that can be produced as an element-wise vector sum from the arc capacities.
  We further include for each such vector the set of all legal vector assignments to arcs in the cut which produce that vector.
\end{definition}

For instance, in Figure \ref{fig:cut_capacities} the enriched Minkowski sum of each point on the RHS of the two equations (corresponding to the two cuts) will have each vector $v$ associated with the set of all pairs of vectors from $a_1$ and $a_3$ or $a_2$ and $a_3$, respectively, that sum to $v$.
Note that the Enriched Minkowski sum gives the set of all local flows over a cut indexed by their flow values.

We consider a generalization of intersection for enriched Minkowski sums of cuts.
This generalization finds flow values in the capacities of a collection of cuts, but only so long as there is a consistent assignment of flow values to arcs which produces the desired flow value across each compared cut---i.e., provided no cuts require double (or more) booking an arc.
We show that taking this intersection over the enriched Minkowski sums of all cuts returns the region of feasible flows along with all realizations of flows that produce any desired flow value.

\begin{definition} [Enriched Intersection] \label{def:EnrichInt}
  The \emph{enriched intersection} of enriched Minkowski sums over a collection of cuts $\mathcal{C}$ is the intersection of the Minkowski sum of capacities over each cut, along with, for each vector $v$ in the intersection of Minkowski sums, the set of gluings of consistent local flows with flow value $v$ for each $c\in \mathcal{C}$.
\end{definition}

We will use intersection in place of enriched intersection and Minkowski sum (or even simply sum) in place of enriched Minkowski when the context is clear.
Likewise, when it is clear that we are taking the sum of a set of vectors, we will use $+$ rather than $+_M$ for the Minkowski sum.
We also use the $\cap$ symbol for the general intersection when it is clear that we are taking the general intersection of enriched Minkowski sums.

\section{Mutual Capacity: Towards an MFMC Result} \label{sec:mutual_capacity}

\newsavebox{\myboxa}
    \begin{lrbox}{\myboxa}{%
        \begin{tikzpicture}[scale=0.75,
roundnode/.style={circle, draw=green!30, fill=green!5, very thick, minimum size=7mm, scale=.7},
roundnode2/.style={circle, draw=black!100, fill=green!0, thick, minimum size=7mm, scale=.7},
squarednode/.style={rectangle, draw=red!60, fill=red!5, very thick, minimum size=5mm},
 el/.style = {inner sep=2pt, align=left, scale=.7},
]
        \centering

        \draw[step=1cm,gray,very thin] (0,0) grid (3.9,3.9);

        \foreach \x in {0,1,2,3,4}
            \draw (\x cm,1pt) -- (\x cm,-1pt) node[anchor=north] {$\x$};
        \foreach \y in {0,1,2,3,4}
            \draw (1pt,\y cm) -- (-1pt,\y cm) node[anchor=east] {$\y$};

        \draw[thick,->] (0,0) -- (4.5,0) node[anchor=north west] {x};
        \draw[thick,->] (0,0) -- (0,4.5) node[anchor=south east] {y};

        \draw[very thick] (0,0)--(1,0)--(0,1)--(0,0);

        \fill[blue] (0,0)--(1,0)--(0,1)--(0,0);

        \node at (5,2) {\large $+$};
        \end{tikzpicture}
    }
    \end{lrbox}

\newsavebox{\myboxb}
\begin{lrbox}{\myboxb}{%
\begin{tikzpicture}[scale=0.75,
roundnode/.style={circle, draw=green!30, fill=green!5, very thick, minimum size=7mm, scale=.7},
roundnode2/.style={circle, draw=black!100, fill=green!0, thick, minimum size=7mm, scale=.7},
squarednode/.style={rectangle, draw=red!60, fill=red!5, very thick, minimum size=5mm},
 el/.style = {inner sep=2pt, align=left, scale=.7},
]
\centering
\draw[step=1cm,gray,very thin] (0,0) grid (3.9,3.9);

\foreach \x in {0,1,2,3,4}
   \draw (\x cm,1pt) -- (\x cm,-1pt) node[anchor=north] {$\x$};
\foreach \y in {0,1,2,3,4}
    \draw (1pt,\y cm) -- (-1pt,\y cm) node[anchor=east] {$\y$};

\draw[thick,->] (0,0) -- (4.5,0) node[anchor=north west] {x};
\draw[thick,->] (0,0) -- (0,4.5) node[anchor=south east] {y};

\draw[very thick] (0,0)--(1,0)--(1,1)--(0,1)--(0,0);

\fill[cyan] (0,0)--(1,0)--(1,1)--(0,1)--(0,0);

\node at (5,2) {\large $=$};
\end{tikzpicture}
}
\end{lrbox}

\newsavebox{\myboxc}
\begin{lrbox}{\myboxc}{%
\begin{tikzpicture}[scale=0.75,
roundnode/.style={circle, draw=green!30, fill=green!5, very thick, minimum size=7mm, scale=.7},
roundnode2/.style={circle, draw=black!100, fill=green!0, thick, minimum size=7mm, scale=.7},
squarednode/.style={rectangle, draw=red!60, fill=red!5, very thick, minimum size=5mm},
 el/.style = {inner sep=2pt, align=left, scale=.7},
]
\centering

\draw[step=1cm,gray,very thin] (0,0) grid (3.9,3.9);

\foreach \x in {0,1,2,3,4}
   \draw (\x cm,1pt) -- (\x cm,-1pt) node[anchor=north] {$\x$};
\foreach \y in {0,1,2,3,4}
    \draw (1pt,\y cm) -- (-1pt,\y cm) node[anchor=east] {$\y$};

\draw[thick,->] (0,0) -- (4.5,0) node[anchor=north west] {x};
\draw[thick,->] (0,0) -- (0,4.5) node[anchor=south east] {y};

\draw[very thick] (0,0)--(2,0)--(2,1)--(1,2)--(0,2)--(0,0);

\fill[magenta] (0,0)--(2,0)--(2,1)--(1,2)--(0,2)--(0,0);
\end{tikzpicture}
}
\end{lrbox}

\newsavebox{\myboxd}
\begin{lrbox}{\myboxd}{%
\begin{tikzpicture}[scale=0.75,
roundnode/.style={circle, draw=green!30, fill=green!5, very thick, minimum size=7mm, scale=.7},
roundnode2/.style={circle, draw=black!100, fill=green!0, thick, minimum size=7mm, scale=.7},
squarednode/.style={rectangle, draw=red!60, fill=red!5, very thick, minimum size=5mm},
 el/.style = {inner sep=2pt, align=left, scale=.7},
]
\centering

\draw[step=1cm,gray,very thin] (0,0) grid (3.9,3.9);

\foreach \x in {0,1,2,3,4}
   \draw (\x cm,1pt) -- (\x cm,-1pt) node[anchor=north] {$\x$};
\foreach \y in {0,1,2,3,4}
    \draw (1pt,\y cm) -- (-1pt,\y cm) node[anchor=east] {$\y$};

\draw[thick,->] (0,0) -- (4.5,0) node[anchor=north west] {x};
\draw[thick,->] (0,0) -- (0,4.5) node[anchor=south east] {y};

\draw[very thick] (0,0)--(3,0)--(3,1)--(2,2)--(0,2)--(0,0);

\fill[purple] (0,0)--(3,0)--(3,1)--(2,2)--(0,2)--(0,0);

\node at (5,2)  {\large  $\cap$};

\end{tikzpicture}
}
\end{lrbox}

\newsavebox{\myboxe}
\begin{lrbox}{\myboxe}{%
\begin{tikzpicture}[scale=0.75,
roundnode/.style={circle, draw=green!30, fill=green!5, very thick, minimum size=7mm, scale=.7},
roundnode2/.style={circle, draw=black!100, fill=green!0, thick, minimum size=7mm, scale=.7},
squarednode/.style={rectangle, draw=red!60, fill=red!5, very thick, minimum size=5mm},
 el/.style = {inner sep=2pt, align=left, scale=.7},
]
\centering

\draw[step=1cm,gray,very thin] (0,0) grid (3.9,3.9);

\foreach \x in {0,1,2,3,4}
   \draw (\x cm,1pt) -- (\x cm,-1pt) node[anchor=north] {$\x$};
\foreach \y in {0,1,2,3,4}
    \draw (1pt,\y cm) -- (-1pt,\y cm) node[anchor=east] {$\y$};

\draw[thick,->] (0,0) -- (4.5,0) node[anchor=north west] {x};
\draw[thick,->] (0,0) -- (0,4.5) node[anchor=south east] {y};

\draw[very thick] (0,0)--(2,0)--(2,2)--(1,3)--(0,3)--(0,0);

\fill[green] (0,0)--(2,0)--(2,2)--(1,3)--(0,3)--(0,0);

\node at (5,2) {\large $=$};
\end{tikzpicture}
}
\end{lrbox}

\newsavebox{\myboxf}
\begin{lrbox}{\myboxf}{%
\begin{tikzpicture}[scale=0.75,
roundnode/.style={circle, draw=green!30, fill=green!5, very thick, minimum size=7mm, scale=.7},
roundnode2/.style={circle, draw=black!100, fill=green!0, thick, minimum size=7mm, scale=.7},
squarednode/.style={rectangle, draw=red!60, fill=red!5, very thick, minimum size=5mm},
 el/.style = {inner sep=2pt, align=left, scale=.7},
]
\centering

\draw[step=1cm,gray,very thin] (0,0) grid (3.9,3.9);

\foreach \x in {0,1,2,3,4}
   \draw (\x cm,1pt) -- (\x cm,-1pt) node[anchor=north] {$\x$};
\foreach \y in {0,1,2,3,4}
    \draw (1pt,\y cm) -- (-1pt,\y cm) node[anchor=east] {$\y$};

\draw[thick,->] (0,0) -- (4.5,0) node[anchor=north west] {x};
\draw[thick,->] (0,0) -- (0,4.5) node[anchor=south east] {y};

\draw[very thick] (0,0)--(2,0)--(2,2)--(0,2)--(0,0);

\fill[teal] (0,0)--(2,0)--(2,2)--(0,2)--(0,0);
\end{tikzpicture}
}
\end{lrbox}

\newcommand{\smallwedge}{\mathrel{\text{\raisebox{0.25ex}{\scalebox{0.8}{$\wedge$}}}}}
\newcommand{\smallvee}{\mathrel{\text{\raisebox{0.25ex}{\scalebox{0.8}{$\vee$}}}}}

Our goal is to define a capacity restriction to replace the straight intersection of cut capacities as used by Krishnan \cite{Kr20??final} which allows us to reduce (or eliminate) the duality gap.
To this end, we first introduce a tighter notion of combining capacities of two cuts, and then generalize to higher number of cuts.

\subsection{Capacity of Two Cuts: Pairwise Capacity} \label{ssec:twocutcap}

We first illustrate the tighter notion of combining capacities of two cuts on a simple example.

\begin{example} \label{duality_gap_example}
Consider the $2$-commodity network introduced in Figure \ref{fig:MultiComNetworkEx}.
Naming the arcs as $a_1=(s,v_2), a_2=(v_2,t)$, and $a_3=(s,t)$, we get two cuts $C_1=\{a_1,a_3\}$ and $C_2=\{a_2,a_3\}$ with capacities shown in Figure \ref{fig:cut_capacities}.

\begin{figure}[ht!]
  \centering
  \includegraphics[scale=1.1]{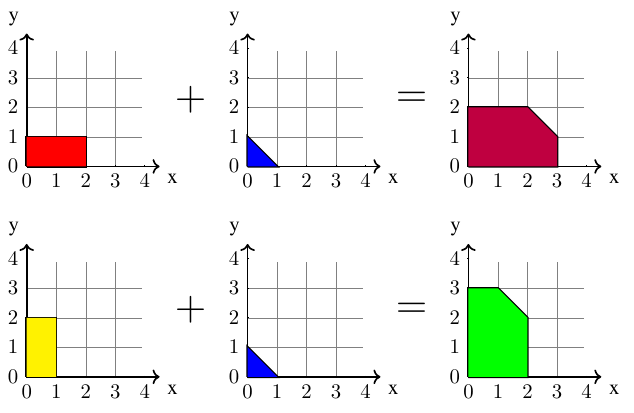}
  \caption{Capacities of the two cuts $C_1, C_2$ in Figure \ref{fig:MultiComNetworkEx}.
  $U(a_1)+U(a_3)=U(C_1)$ (first row) and $U(a_2)+U(a_3)=U(C_2)$ (second row).}
  \label{fig:cut_capacities}
\end{figure}

\begin{definition}[Pairwise Capacity]
  The \emph{pairwise capacity} of two sets of arcs $A=\{a_i\}$ and $B=\{b_j\}$ is defined as $U_P(A,B) = v(A \cap B) + (v(A \setminus B) \cap v(B \setminus A))$.
In words, we take the capacity of the set of arcs in the intersection of the two sets, then add the intersection of the capacity of the arcs unique to each set.
\end{definition}

We refer to pairwise capacity also as \textit{two-cut mutual capacity}, as it is a special case of mutual capacity.
For instance, the pairwise capacity $U_P(C_1,C_2)$ for the cuts in our example requires us to find $C_1 \cap C_2 = \{a_3\}$.
Also note that $C_1$ has the unique arc $a_1$ and $C_2$ has the unique arc $a_2$.
Hence we get $U_P(C_1,C_2) = v(\{a_3\}) + v(a_1) \cap v(a_2)$, as shown in Figure \ref{fig:PairwiseCap}.

\begin{figure} [!ht]
  \centering
  \includegraphics[scale=1.1]{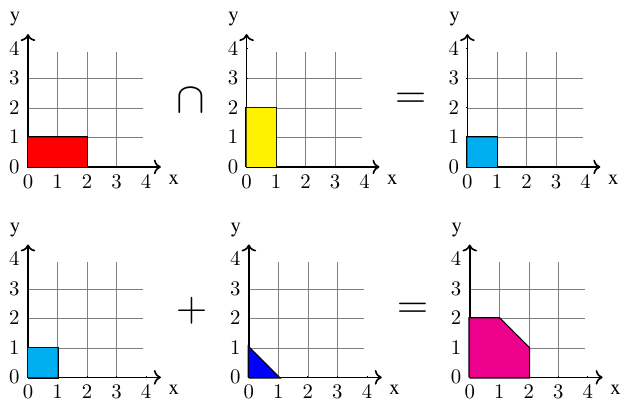}
  \caption{
    Pairwise capacity of two cuts $C_1, C_2$ in Figure \ref{fig:MultiComNetworkEx}.
    First row shows $v(a_1) \cap v(a_2)$.
    Second row shows the pairwise capacity as $(v(a_1) \cap v(a_2)) + v(\{a_3\})$.
    Note that $C_1 \cap C_2 = \{a_3\}$.
  }
  \label{fig:PairwiseCap}
\end{figure}
\end{example}

We note that $A+(B\cap C) \subseteq (A+B)\cap (A+C)$ for the Minkowski sum, but the two sets are not equal in general.
We also note that for a set of arcs $A$ partitioned into sets $B$ and $C$, $v(A)=v(B)+v(C)$.
This is true because any flow $a$ in $v(A)$ can be written as $a=b+c$ where $b \in v(B)$ and $c\in v(C)$, and any flow that can be written as a sum of a flow in $v(B)$ plus one in $v(C)$ lies in $v(A)$.

It follows that $U_P(A,B) \subseteq v(A) \cap v(B)$.
Applied to a family of cuts $\mathcal{C}$, we can then say that $\bigcap \limits_{c_i, c_j \in \mathcal{C}} U_P(c_i,c_j) \subseteq \bigcap \limits_{c_l \in \mathcal{C}} C_l$.
And since the total capacity is given by $ \bigcap \limits_{c_l \in \mathcal{C}} C_l$, we can say that $\bigcap \limits_{c_i, c_j \in \mathcal{C}} U_{P}(c_i,c_j)$ is in general a tighter constraint than that given by Ghrist and Krishnan \cite{GhKr2013}.

We demonstrate that total capacity may include some flow values that lack a feasible realization in Figures \ref{Total_Cap_Gap} and \ref{MCMFExTotalCapsGap}.
For the network in Figure \ref{fig:MultiComNetworkEx}, the total capacity of the two cuts $C_1, C_2$ is shown in Figure \ref{Total_Cap_Gap}.
At the same time, the corresponding pairwise capacity shown in Figure \ref{fig:PairwiseCap} is strictly smaller than the total capacity.

We observe that flow $(2,2)$ is contained in the total capacity.
To achieve a flow of $(2,2)$ through the network, cut $C_1$ requires a flow of $(2,1)$ through $a_1$ and $(0,1)$ through $a_3$.
At the same time, cut $C_2$ requires a flow of $(1,2)$ through $a_2$ and $(1,0)$ through $a_3$.
Since $a_3$ cannot take a flow of $(0,1)$ and $(1,0)$ at the same time, the network cannot transport $(2,2)$.

\begin{figure}[ht!]
  \centering
  \includegraphics[scale=1.1]{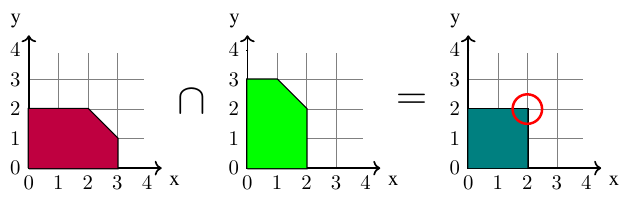}
  \caption{
    Capacities  $U(C_1)$ (left), $U(C_2)$ (center), and total capacity (right).
    Figure \ref{fig:cut_capacities} shows the computation of $U(C_1)$ and $U(C_2)$.
    Note the vector $(2,2)$ lies in total capacity.
    }
    \label{Total_Cap_Gap}
\end{figure}

\begin{figure}[ht!]
  \centering
  \includegraphics[scale=1.25]{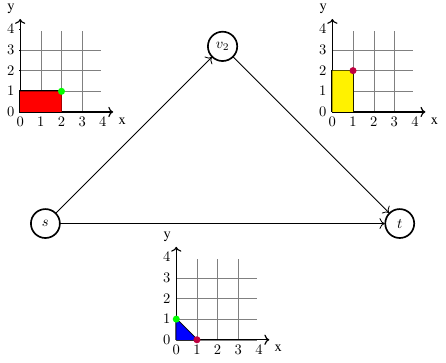}
  \caption{Demonstration of gap in total capacity.
  Arc flows needed to realize the flow $(2,2)$, which is contained in the total capacity as shown in Figure \ref{Total_Cap_Gap}, are shown as green dots for cut $C_1$ and pink dots for cut $C_2$.
  But both flows (green and pink dots) cannot be \emph{simultaneously} realized in $a_3$.
  }
  \label{MCMFExTotalCapsGap}
\end{figure}

The example demonstrates that pairwise capacity could avoid gaps allowed by total capacity.
To prove that pairwise capacity still bounds multicommodity flows, we must show that any valid flow must lie within the given capacity bounds.
The following lemma presents this result.

\begin{lemma} \label{lem:inclusion}
  Given a pair of cuts $C_i, C_j$, any valid flow must fall within
  \[ U_P(C_i,C_j) = v(C_i \cap C_j) + (v(C_i \setminus C_j) \cap v(C_j \setminus C_i)),\]
  i.e., the sum of the values of all arcs common to the cuts in the collection plus the intersection of cut values not included in the intersection.
\end{lemma}

\begin{proof}
  For any pair of cuts, flow can be partitioned into a portion through the common arcs and flow through the remaining arcs.
  Since the flow can assign only one value to each of the common arcs, we must be able to write the value of flow as $a+b$ and $a+c$ (in standard vector addition) where $a$ is in $v (\text{the common arcs})$, $b$ is in $v(\text{the arcs unique to } C_i)$, and $c$ is in $v(\text{the arcs unique to } C_j)$.
  Since the flow value across each cut must be equal, $b=c$.
  Hence $a\in v(C_i\cap C_j)$ and $b=c\in (v(C_i\setminus C_j)\cap v(C_j \setminus C_i))$.
  Thus the feasible flow values through $C_i$ and $C_j$ must fall within $U_{P} (C_i,C_j)$.
\end{proof}

We further note that since the feasible flow values through any pair of cuts $(C_i,C_j)$ is bounded by $U_{P} (C_i,C_j)$, the feasible flow values through the network must be bounded by $\bigcap \limits_{i,j} U_{P} (C_i,C_j)$ taken over all pairs of cuts.

\subsubsection{Duality Gap of Pairwise Capacity} \label{sssec:Gap_example}

While pairwise capacity gives us a tighter approximation than total capacity, it may still not capture the exact feasible region.
We present an example that illustrates this gap.

\begin{example} \label{ex:Gap_example}
  
We present in Figure \ref{ExNet} a network for which taking the intersection over all $s$--$t$ cuts gives a duality gap with the true maximum flow.

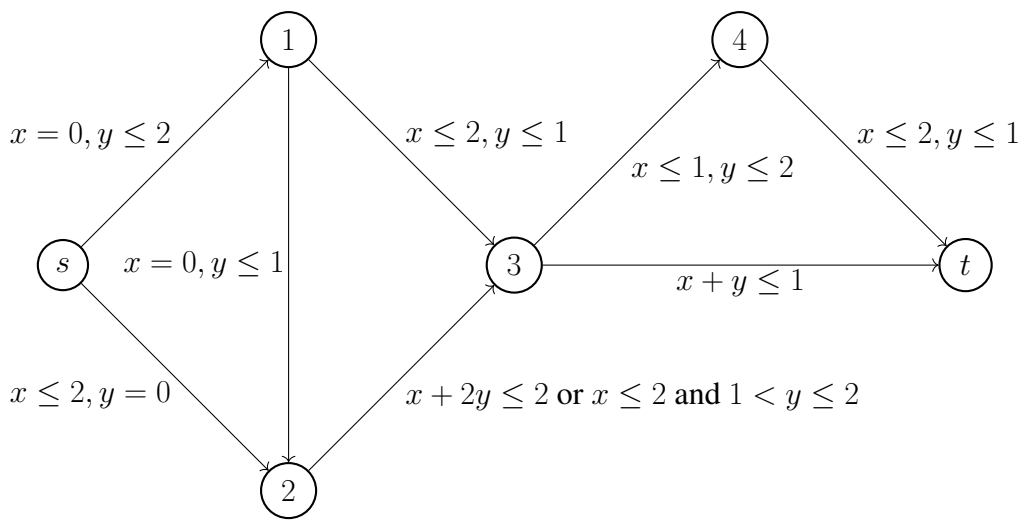
\begin{figure}[ht!]
    \Large
    \centering
    \begin{tikzpicture}[
roundnode/.style={circle, draw=green!30, fill=green!5, very thick, minimum size=7mm, scale=.7},
roundnode2/.style={circle, draw=black!100, fill=green!0, thick, minimum size=7mm, scale=.7},
squarednode/.style={rectangle, draw=red!60, fill=red!5, very thick, minimum size=5mm},
 el/.style = {inner sep=2pt, align=left, scale=.7},
]
\foreach \i in {3}{

\node[roundnode2] at (0,0) (s) {$s$};
\node[roundnode2] at (1*\i,1*\i) (v1) {$1$};
\node[roundnode2] at (1*\i,-1*\i) (v2) {$2$};
\node[roundnode2] at (2*\i,0) (v3)  {$3$};
\node[roundnode2] at (3*\i,1*\i) (v4)  {$4$};
\node[roundnode2] at (4*\i,0) (t)  {$t$};
}

\path
 (s) edge [->] node[el,above left] {$x=0, y\leq 2$} (v1)
 (s) edge [->] node[el,below left] {$x\leq 2, y=0$} (v2)
 (v1) edge [->] node[el,left] {$x=0, y\leq 1$} (v2)
 (v1) edge [->] node[el,above right] {$x\leq 2, y\leq 1$} (v3)
 (v2) edge [->] node[el, below right] {$x+2y\leq 2 \text{ or } x\leq 2 \text{ and } 1 < y \leq 2$} (v3)
 (v3) edge [->] node[el, below right] {$x\leq 1, y\leq 2$} (v4)
 (v3) edge [->] node[el,below] {$x+y\leq 1$} (t)
 (v4) edge [->] node[el,above right] {$x\leq 2, y\leq 1$} (t);
    
    \end{tikzpicture}
    \caption{A network for which pairwise capacity reduces but does not eliminate the duality gap.}
    \label{ExNet}
\end{figure}

The intersection over all cuts gives the capacity $(x \leq 2, y \leq 2)$ for this network.
However, we make the following observations.
\begin{itemize}
    \item We cannot get any $x$ to node $1$, so we can rewrite the capacity of $(1,3)$ as $(x=0, y \leq 1)$.
    \item We cannot get $(x=2, y=2)$ to node $2$, but we can get up to two units of $x$ and one unit of $y$, so the capacity can be rewritten as $(x+2y \leq 2)$.
\end{itemize}
We can determine the true feasible region using these observations as shown in Figure \ref{fullcap}. For instance, it is impossible to find a flow with flow value $(1,2)$.
Since $(s,1)$ can transport only commodity $y$ and $(s,2)$ can transport only commodity $x$, we must route two units of $y$ through $(s,1)$.
Likewise, we must route one unit of $x$ through $(s,2)$.
Then since inflow equals outflow, we must send one unit of $y$ through each of $(1,2)$ and $(1,3)$.
But this means that we must send one unit each of $x$ and $y$ along $(2,3)$.
However, this is outside the capacity of $(2,3)$.
Hence $(1,2)$ is outside the feasible region of flow values even though it lies in all pairwise capacities in this network.

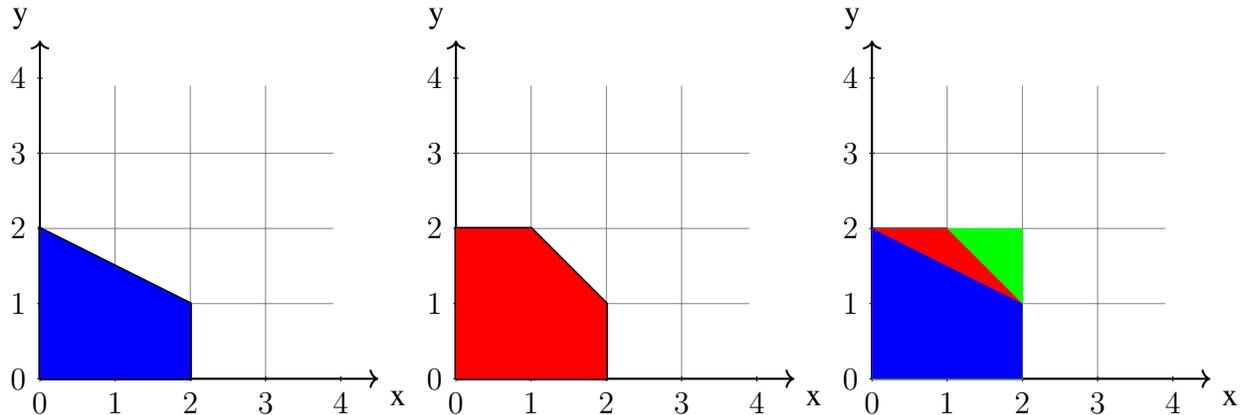
\begin{figure}[ht!]
\begin{minipage}{.33\textwidth}

\centering
\begin{tikzpicture}[scale=1,
roundnode/.style={circle, draw=green!30, fill=green!5, very thick, minimum size=7mm, scale=.7},
roundnode2/.style={circle, draw=black!100, fill=green!0, thick, minimum size=7mm, scale=.7},
squarednode/.style={rectangle, draw=red!60, fill=red!5, very thick, minimum size=5mm},
 el/.style = {inner sep=2pt, align=left, scale=.7},
]
\centering

\draw[step=1cm,gray,very thin] (0,0) grid (3.9,3.9);

\foreach \x in {0,1,2,3,4}
   \draw (\x cm,1pt) -- (\x cm,-1pt) node[anchor=north] {$\x$};
\foreach \y in {0,1,2,3,4}
    \draw (1pt,\y cm) -- (-1pt,\y cm) node[anchor=east] {$\y$};

\draw[thick,->] (0,0) -- (4.5,0) node[anchor=north west] {x};
\draw[thick,->] (0,0) -- (0,4.5) node[anchor=south east] {y};

\draw [very thick] (0,0)--(0,2)--(2,1)--(2,0)--(0,0);
\fill[blue] (0,0)--(0,2)--(2,1)--(2,0)--(0,0);

\end{tikzpicture}
\end{minipage}
\begin{minipage}{.33\textwidth}

\centering
\begin{tikzpicture}[scale=1,
roundnode/.style={circle, draw=green!30, fill=green!5, very thick, minimum size=7mm, scale=.7},
roundnode2/.style={circle, draw=black!100, fill=green!0, thick, minimum size=7mm, scale=.7},
squarednode/.style={rectangle, draw=red!60, fill=red!5, very thick, minimum size=5mm},
 el/.style = {inner sep=2pt, align=left, scale=.7},
]
\centering

\draw[step=1cm,gray,very thin] (0,0) grid (3.9,3.9);

\foreach \x in {0,1,2,3,4}
   \draw (\x cm,1pt) -- (\x cm,-1pt) node[anchor=north] {$\x$};
\foreach \y in {0,1,2,3,4}
    \draw (1pt,\y cm) -- (-1pt,\y cm) node[anchor=east] {$\y$};

\draw[thick,->] (0,0) -- (4.5,0) node[anchor=north west] {x};
\draw[thick,->] (0,0) -- (0,4.5) node[anchor=south east] {y};

\draw [very thick] (0,0)--(0,2)--(1,2)--(2,1)--(2,0)--(0,0);
\fill[red] (0,0)--(0,2)--(1,2)--(2,1)--(2,0)--(0,0);

\end{tikzpicture}
\end{minipage}
\begin{minipage}{.33\textwidth}
\begin{tikzpicture}[scale=1,
roundnode/.style={circle, draw=green!30, fill=green!5, very thick, minimum size=7mm, scale=.7},
roundnode2/.style={circle, draw=black!100, fill=green!0, thick, minimum size=7mm, scale=.7},
squarednode/.style={rectangle, draw=red!60, fill=red!5, very thick, minimum size=5mm},
 el/.style = {inner sep=2pt, align=left, scale=.7},
]
\centering

\draw[step=1cm,gray,very thin] (0,0) grid (3.9,3.9);

\foreach \x in {0,1,2,3,4}
   \draw (\x cm,1pt) -- (\x cm,-1pt) node[anchor=north] {$\x$};
\foreach \y in {0,1,2,3,4}
    \draw (1pt,\y cm) -- (-1pt,\y cm) node[anchor=east] {$\y$};

\draw[thick,->] (0,0) -- (4.5,0) node[anchor=north west] {x};
\draw[thick,->] (0,0) -- (0,4.5) node[anchor=south east] {y};

\draw (0,0)--(0,2)--(1,2)--(2,1)--(2,0)--(0,0);
\draw (0,2)--(1,2);

\fill[green] (0,0)--(0,2)--(2,2)--(2,0)--(0,0);
\fill[red] (0,0)--(0,2)--(1,2)--(2,1)--(2,0)--(0,0);
\fill[blue] (0,0)--(0,2)--(2,1)--(2,0)--(0,0);

\end{tikzpicture}
\end{minipage}
\caption{
  Feasible regions: true feasible region in blue (left), the region from pairwise capacities in red (middle), and the total capacity in green shown along with the other two regions for easy comparison (right).
}
\label{fullcap}
\end{figure}

To find the mutual capacity of the cuts $\{(s,2),(1,2),(1,3)\}$ and $\{(s,1),(s,2)\}$, we take the capacity of the common arc $(s,2)$, i.e., ($x \leq2, y=0$), and add the intersection of the sum of capacities of arcs unique to each cut, i.e.,
\[ (x=0, y \leq 2) \cap ( (x=0, y \leq 1) + (x \leq 2, y \leq 1) ) = (x=0, y \leq 2), \]
for a total bound of $(x \leq 2, y \leq 2)$.
Hence for the portion of the network with nodes $\{s,1,2,3\}$, we get an approximation of the actual capacity given by $(x \leq 2, y \leq 2)$.

For the portion of the network with nodes $\{3,4,t\}$, we get the region in the center in Figure \ref{fullcap} (this calculation is identical to the one shown in Figure \ref{fig:PairwiseCap}).
As this region is a subset of $(x \leq 2, y \leq 2)$, the intersection over all two-way mutual capacities gives us the same region shown in the center of Figure \ref{fullcap}.
Note how the true capacity of the network is contained in the approximation from pairwise capacities, which in turn is contained in the total capacity.
\end{example}

\subsection{Mutual Capacity of Fully Disjoint Networks} \label{ssec:DisjointMutualCap}

Before we consider the generalization of pairwise capacity to the case of combining capacities of multiple cuts in the most general setting, we present the result for a class of networks with specific structure.
The reader will find our arguments for this special case easier to follow before transitioning to the general case.
We say a network $N$ is \emph{fully disjoint} if all of its $s$--$t$ paths are node disjoint except at $s$ and $t$.
We illustrate some new definitions and results on the network in Figure \ref{fig:restriction_network}, where $a_{i,j}$ has capacity $c_{i,j}$.

\begin{figure}[htp!]
\centering
\begin{tikzpicture}[
roundnode/.style={circle, draw=green!30, fill=green!5, very thick, minimum size=7mm, scale=.7},
roundnode2/.style={circle, draw=black!100, fill=green!0, thick, minimum size=7mm, scale=.7},
squarednode/.style={rectangle, draw=red!60, fill=red!5, very thick, minimum size=5mm},
 el/.style = {inner sep=2pt, align=left, scale=.7},
]
\node[roundnode2] (v1) {$s$};
\node[roundnode2] (v3) [right=of v1] {$v_3$};
\node[roundnode2] (v2) [above=of v3] {$v_2$};
\node[roundnode2] (v4) [below=of v3] {$v_4$};
\node[roundnode2] (v5) [right=of v2] {$v_5$};
\node[roundnode2] (v6) [right=of v3] {$v_6$};
\node[roundnode2] (v7) [right=of v4] {$v_7$};
\node[roundnode2] (v8) [right=of v6] {$t$};

\path
 (v1) edge [->] node[el,above left] {$a_{1,1}$} (v2)
 (v1) edge [->] node[el,below] {$a_{2,1}$} (v3)
 (v1) edge [->] node[el, below left] {$a_{3,1}$} (v4)
 (v2) edge [->] node[el,above] {$a_{1,2}$} (v5)
 (v3) edge [->] node[el,below] {$a_{2,2}$} (v6)
 (v4) edge [->] node[el, above] {$a_{3,2}$} (v7)
 (v5) edge [->] node[el,above right] {$a_{1,3}$} (v8)
 (v6) edge [->] node[el,below] {$a_{2,3}$} (v8)
 (v7) edge [->] node[el, below right] {$a_{3,3}$} (v8);

\end{tikzpicture}
\caption{Example network for restrictions.}
\label{fig:restriction_network}
\end{figure}
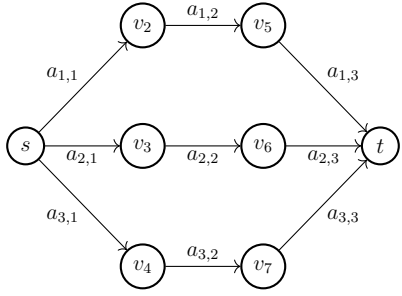

To use mutual capacity to find the true capacity of a fully disjoint network, we develop a few more nuanced results.
First, we expand our idea of mutual capacity.

\begin{definition}[Restriction]
  A \textit{restriction} on a network is a constraint on the flow capacities that can be written as a sum over all possible intersections of arc capacities.
\end{definition}

All cuts are restrictions, as are the mutual capacities of pairs of cuts.

\begin{definition}[Compatible]
  Two restrictions are said to be \emph{compatible} if their capacities differ in exactly one summand.
\end{definition}

In the network in Figure \ref{fig:restriction_network} the cuts $C_{1,1}=\{a_{1,1},a_{2,1},a_{3,1}\}$ and $C_{1,2}=\{a_{1,2},a_{2,1},a_{3,1}\}$ are compatible as their capacities differ only along the first path.
Likewise, $U_P(C_{1,1},C_{1,2}) = c_{2,1} + c_{3,1} + (c_{1,1} \cap c_{1,2})$ is compatible with $C_{1,3} = \{a_{1,3},a_{2,1},a_{3,1}\}$ since $U(C_{1,3}) = c_{2,1} + c_{3,1} + c_{1,3}$, which differs from $U_P(C_{1,1},C_{1,2})$ in exactly one term ($c_{1,3}$ rather than $(c_{1,1}\cap c_{1,2})$).
However,  $C_{1,2}$ and $C_{2,2} =\{a_{1,1},a_{2,2},a_{3,1}\}$ are not compatible since they each have two unique arcs and hence two unique elements in their capacities.

The pairwise capacity of a pair of compatible restrictions is a generalization of the pairwise capacity calculation for a pair of cuts.
Let $R_1=\{a_1,\dots,a_n\}$ and $R_2=\{a_1,\dots,a_{n-1},b_n\}$. Then $U_P(R_1,R_2)=U(\{a_1,\dots,a_{n-1}\})+(U(a_n) \cap U(b_n))$.
For instance, we noted above that $U_P(C_{1,1},C_{1,2})$ is compatible with $C_{1,3}$.
Then $U_P(U_P(C_{1,1},C_{1,2}),C_{1,3})$ is given by $c_{2,1}+c_{3,1}$, the common capacities, plus $c_{1,3} \cap (c_{1,1} \cap c_{1,2})$.
In fact, we show that all valid flows must satisfy such mutual capacities.

\begin{lemma} \label{lem:validflowcomprest}
  Any valid flow in a fully disjoint network satisfies mutual capacities of all compatible restrictions.
\end{lemma}

\begin{proof}
  Suppose $f$ is a valid flow through a fully disjoint network.
  Then by Lemma \ref{lem:inclusion}, valid flows must abide by restrictions from pairs of compatible cuts.
  Now suppose that $f$ must abide by two compatible restrictions $R_1$ and $R_2$.
  Since they differ in exactly one summand, $f$ must take the same value on the region of the network represented by common summands and thus the same value on the unique summand for each restriction.
  Thus the value of $f$ on the unique summand in each restriction must lie in the intersection of the capacities of those summands.
  Hence the value of $f$ is in the restriction formed by taking the sum of common summands in $R_1$ and $R_2$ plus the intersection of the value of their unique summands.
\end{proof}

We now prove that the intersection over the values of all mutual capacities of cuts and all compatible restrictions gives the set of achievable flow values in the case of fully disjoint networks.

\begin{theorem} \label{thm:disjointMFMC}
  Let $N$ be a fully disjoint network with $p$ $s$--$t$ paths.
  We assume that extra edges are introduced into each path except the longest one with capacity equal to that of the final arc in that path so that all paths have the same length $\ell$.
  Let $a_{i,j}$ denote the $j$-th arc along the $i$-th path.
  Then the intersection over the values of all pairwise capacities of cuts and all compatible restrictions gives exactly the set of all achievable flow values: $\sum \limits_{i=1}^{p} \bigcap \limits_{j=1}^{\ell} U(a_{i,j})$.
\end{theorem}

\begin{proof} 
  We first note that by Lemma \ref{lem:validflowcomprest}, the set of achievable flows must be a subset of such an intersection.
  Hence it remains to show that such an intersection is a subset of the set of achievable flows.
  It will suffice to show that an intersection over some subset of the pairwise capacities equals the set of all achievable flow values.

  We number the $s$--$t$ paths from $1$ to $p$.
  Since all paths are fully disjoint, each collection of arcs containing exactly one arc from each path forms a cut.
  All cuts of our interest have this form.
  
  For any cut $C_1$ using the first arc in the first path, there is a set of compatible cuts $C_j$ each identical to $C_1$ except that $C_j$ uses arc $j$ from path 1 rather than arc 1.
  This is a set of compatible cuts with capacity $\texttt{Cap}+\cap_{j=1}^{\ell} c_j$, where $\texttt{Cap}$ is the sum of capacities from all arcs except the arc from path 1, and $c_j$ is the capacity of arc $j$ from path 1.
  
  Then each cut has an associated, tighter restriction formed by replacing the capacity of an arc from path 1 with the intersection of capacities over that path.
  One can think of this step as replacing each cut with this tighter restriction, essentially reducing path 1 to a single arc with capacity equal to the intersection over capacities in the original path 1.
  
  Since all cuts now take the same value on path 1, we can repeat this process over the remaining paths.
  At each step, we generate restrictions that replace the capacities from the current path with the intersection over capacities from the same path.
  Ultimately, this gives us our tightest restriction with capacity $\sum \limits_{i=1}^{p} \bigcap \limits_{j=1}^{\ell} U(a_{i,j})$, i.e., for each path take the intersection of capacities of arcs over that path, then sum the values of those intersections.
\end{proof}

In our example above, we first partition cuts by equivalence relation into sets so that cuts in the one set pass through the same arcs in paths 2 and 3.
For instance, one set is $\{\{a_{1,1},a_{2,1},a_{3,1}\}$, $\{a_{1,2},a_{2,1},a_{3,1}\}$, $\{a_{1,3},a_{2,1},a_{3,1}\}\}$, and another can be $\{\{a_{1,1},a_{2,3},a_{3,1}\}$, $\{a_{1,2},a_{2,3},a_{3,1}\}$, $\{a_{1,3},$ $a_{2,3},a_{3,1}\}\}$.
For each set, we generate a restriction with capacity $\texttt{Cap}=$ sum of capacities of arcs from path $2$  and  $3$ $+ \bigcap c_{1,j}$,  where $c_{1,j}$ is the capacity of arc $a_{1,j}$.

Using these restrictions, we repeat the process along path 2.
So now, with $c_1= \bigcap_j U(a_{1,j})$ our partitioned sets include $\{\{a_{2,1}, a_{3,1}\}$, $\{a_{2,2},a_{3,1}\}$, $\{a_{2,3},a_{3,1}\}\}$, for instance.
Note that the restrictions from our first iteration are uniquely identified by two arcs, since the path 1 component in all restrictions is the same.

Applying the same restriction-generating process to each set, we get a restriction with capacity $\texttt{Cap}=\text{ capacity of arcs from path 3}+\bigcap \limits_{j}{c_{1,j}}+ \bigcap \limits_{j}{c_{2,j}}$.

This set of restrictions is already compatible, as they differ only in the arc each of them include from path 3.
Applying  our restriction-generating process again, we now get a single restriction with capacity $\texttt{Cap}=\bigcap \limits_{j}{c_{1,j}}+ \bigcap \limits_{j}{c_{2,j}} +\bigcap \limits_{j}{c_{3,j}}$.
This is the true capacity, as we can verify independently.

\section{General Mutual Capacity and Duality Result} \label{sec:GenMutualCap}

Our approach to prove Theorem \ref{thm:disjointMFMC} does not generalize to all networks.
For instance, Example \ref{ex:Gap_example} presents a network for which two-cut mutual capacities reduce but do not eliminate the gap.
We need a further generalization to calculate the true feasible region in all cases.

To this end, we use the enriched Minkowski sum and enriched intersection (Definitions \ref{def:EnrichMsum} and \ref{def:EnrichInt}) instead of using Minkowski sums and intersections. 
Then for a single cut, the capacity is the enriched Minkowski sum.
For a set of cuts $\mathcal{C}$, the mutual capacity $U_M(\mathcal{C})$ should capture the collection of flows which can fit through all cuts.
The general duality result should then give that the mutual capacity of a set of cuts is the set of all local flows over the collection of cuts.
We present Lemmas \ref{flow_equality}, \ref{local_gluing}, and \ref{lem:mutualcap} in preparation for the main result in Theorem \ref{thm:MCMFMC}.


\begin{lemma} \label{flow_equality}
  Given a flow $f$, net flow across any cut is equal to the flow value of $f$.
\end{lemma}

\begin{proof}
  Decompose $f$ into separate flows in each commodity in single commodity networks.
  Then for each commodity, net flow is the same over each cut in the network.
  Hence net flow of $f$ is the same over each cut, as net flow of $f$ over a cut is exactly a vector returning the net flow in each commodity.
\end{proof}

\begin{lemma} \label{local_gluing}
  A set of compatible local flows admits a local flow over the union of their arc sets and cut sets.
  \end{lemma}

\begin{proof}
  For each arc in the union we find a contributing local flow which assigns that arc a value, then assign that value to the arc.
  Compatibility guarantees that all local flows which assign that arc a value assign it the same one.
  The flow so admitted, called the gluing of the local flows, is the unique local flow over the union of the arc sets and cut sets which, when restricted to the arc set and cut set of one of the contributing flows, returns that local flow.

  Suppose $\{L_i\}_{i\in I}$ is a set of compatible local flows, and suppose $L$ is the product of the above gluing operation.
  We must show that $L$ is a local flow over the union of arc sets and that $L$ reduces to $L_i$ when restricted to its arc set and cut set. 

  Because the local flows are compatible, each cut in the union of cut sets has the same value, the value each assigns to $e$.
  Furthermore, compatibility guarantees that $L$ does not assign multiple values to any arc.
  Because it assigns a value to an arc if and only if at least one local flow assigns a value, it assigns values to arcs if and only if they are in the union of the local flows' arc sets.
  Thus $L$ is a local flow over the union of the cut sets of the contributing local flows.
  Furthermore, this assignment is uniquely determined.

  Likewise, the gluing condition guarantees that $L$ agrees with each $L_i$ over all arcs to which $L_i$ assigns values, so $L$ restricted to the arc and cut sets of $L_i$ gives exactly the assignment $L_i$.
\end{proof}

\begin{lemma}[Mutual Capacity] \label{lem:mutualcap}
  The mutual capacity $U_M$ of a set of cuts is the set of local flows which can be constructed as the gluing of compatible local flows from each cut in the set of cuts.
\end{lemma}

\begin{proof}
  Suppose $f$ is in the mutual capacity of a set of arcs.
  Then, by definition, $f$ is a local flow over the set of cuts.
  Let $f_i$ be the local flow obtained by restricting $f$ to cut $i$.
  Note that $\{f_i\}$ is a compatible set of local flows which glue to $f$.

  Thus each element of $U_M$ is a local flow in the gluing of local flows over each cut.

  Suppose $f$ is the gluing of a set of local flows over each cut in a collection of cuts.
  Then $f$ is a local flow over the collection of cuts, and hence in $U_M$.

  Thus a local flow is in $U_M$ if and only if it is the gluing of a set of compatible local flows from each cut in the set of cuts.
\end{proof}


\begin{theorem} \label{thm:MCMFMC}
  The set of feasible flows through a network is precisely the mutual capacity of all cuts in the network.
\end{theorem}

\begin{proof}
  We show that a flow is feasible if and only if it is in the mutual capacity of all cuts in the network.

  Suppose a flow $f$ is feasible.
  The restriction $f|_c$  of $f$ to each cut $c \in \mathcal{C}$, the set of cuts in the network, gives a local flow over that cut.
  Then $f$ is the gluing of all $\{f|_c\}_{c \in \mathcal{C}}$ (by Lemma \ref{local_gluing}) , and hence in the mutual capacity of all cuts.

  Suppose $f$ is in the mutual capacity of all cuts.
  Then $f$ respects all capacity constraints, so all that remains to show is that the flow respects conservation.
  For each node $\nu$ and each commodity $\kappa$, take $C_1$ to be a cut with cut set including the in-arcs of $\nu$ for commodity $\kappa$, and take $C_2$ to have the same cut set but including the out-arcs of $\nu$ for commodity $\kappa$ rather than the in-arcs, i.e., the same cut but with $\nu$ moved from the $t$-set to the $s$-set.

  Since the net flow across each cut must be equal by Lemma \ref{flow_equality}, we have that
  \[ \sum_{a \in C_1} \phi(a) = \sum_{b \in C_2} \phi(b),\]
  and hence
  \[ \sum_{a \in \text{ in-arcs of } \nu} \phi(a) + \sum_{c \in C_1 \cap C_2} \phi(c) = \sum_{b \in \text{ out-arcs of } \nu} \phi(b) + \sum_{c \in C_1\cap C_2} \phi(c),\]
  giving
  \[\sum_{a \in \text{ in-arcs of } \nu} \phi(a) = \sum_{b \in \text{ out-arcs of } \nu} \phi(b).\]
  This result holds for each commodity $\kappa \in \{1,\dots,k\}$ and at each node $\nu \in V$.
  Hence flow conservation is respected for each commodity at all nodes, and hence over all commodities.
  Therefor $f$ is a feasible flow.
  
  By Lemma \ref{lem:mutualcap}, this gives that the set of feasible flows is exactly the mutual capacity over all cuts.
\end{proof}

\subsection{Min-Cut Algorithm} \label{ssec:MCAlgo}
We consider computational implications of the tight duality result in Theorem \ref{thm:MCMFMC}.
While identifying the minimum cut in our multicommodity setting is expected to be expensive, we explore situations where computing the mutual capacity may turn out to be more efficient than exhaustive enumeration.

Let $X_i$ be the capacity semimodule assigned to arc $i$, with $x_i$ an arbitrary element of $X_i$.
We consider the enriched Minkowski sum in Definition \ref{def:EnrichMsum}
as a more informative version of the Minkowski sum over a cut-set $\mathcal{C}_j=\{a_\kappa\}$ the function $f_j(x_1,x_2,\cdots,x_m)=\Sigma_{\kappa} x_\kappa$.
In the worst case, this is no more computationally expensive than the Minkowski sum, since one must already calculate all element-wise sums.
We note that there are some special cases where the Minkowski sum is less computationally expensive, most notably in the case of convex polyhedra, but in most cases the expense is comparable.

The intersection of these objects is their intersection as graphs of functions over $\Pi_{i\in I} X_i$, or, equivalently, the set of all  $(x_1,x_2,\cdots,x_m)$ such that all $f_j((x_1,x_2,\cdots,x_m))$ are equal.

Note that this algorithm cannot run in sub-exponential time since a fully connected graph will have $2^{n-2}$ edge-minimal $s$--$t$ cuts, and we must intersect over all of them.
However, note that this is the case for Krishnan's special cases \cite{Kr20??final} as well.
On the other hand, these computations can be done on sparse networks with less expense.

To find a realization of a particular flow value, it will suffice to find an element in the levelset of that value in the intersection of the graphs, or, again equivalently, an equalizer for all cuts.
We illustrate this computation on a practical example.
We finish this section by presenting a class of networks on which our algorithm can run exponentially faster (in a capacity parameter) than brute force search.

\begin{example} \label{ex:duality}

\begin{figure}[ht!]
  \centering
  \includegraphics[scale=1.1]{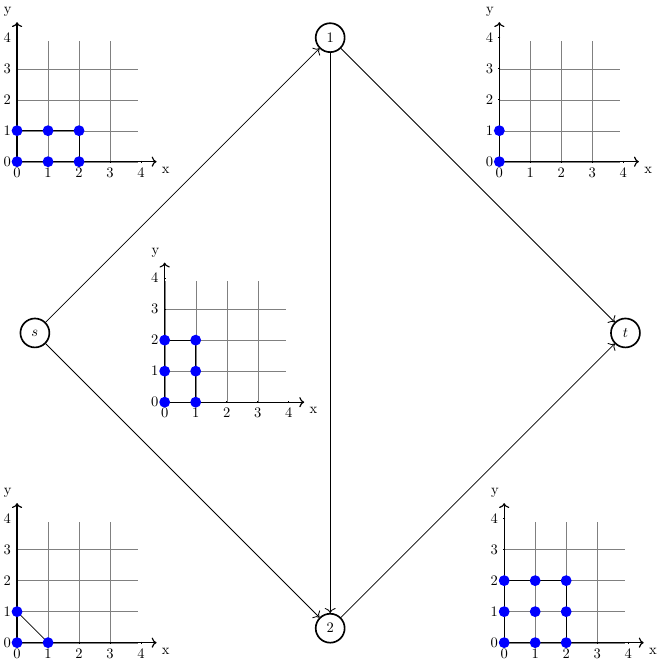}
  \caption{Network considered in Example \ref{ex:duality}.}
  \label{fig:duality}
\end{figure}

Consider a network with capacities given as sets of integer vectors in Figure \ref{fig:duality}.

First we find the local flows over each cut by listing all combinations of flow values that can be assigned to arcs in that cut.
Then we calculate the flow value over the cut for each combination by adding flow values on forward arcs and subtract flow values on backward arcs in  the cut.
This gives the net flow value for each local flow over the cut.

To glue together the local flows over two cuts, we find all combinations of local flows that agree on all arcs the two cuts share (including the net flow as assigned to the edge $e$ from $t$ to $s$).
We then construct the local flows over the collection of cuts by creating for each pair of compatible local flows discovered a local flow that assigns to each arc the value which at least one of the contributing cuts assigns.
For instance, the local flow that assigns $(0,1)$ to $s$--$2$, $(1,0)$ to $1$--$2$, and $(0,1)$ to $1$--$t$ can be glued to the local flow that assigns $(0,1)$ to $s$--$2$ and $(1,1)$ to $s$--$1$ to give a local flow over both cuts assigning $(0,1)$ to $s$--$2$, $(1,1)$ to $s$--$1$, $(1,0)$ to $1$--$2$, and $(0,1)$ to $1$--$t$.

We construct the set of global flows over the network by starting from the first cut, gluing the second cut to it, gluing the third cut to the local flows over the first two cuts, and iterating through all cuts of the graph.
The result, a local flow over all cuts, is precisely the set of feasible flows on the graph.
\end{example}

\begin{example}
  Our algorithm can greatly outperform brute force approaches.
  Consider the network in Figure \ref{ex:multi_com}.
  Note that $(s,v_1)$ has six elements in its capacity.
  Numbering the cuts 1 to 3 from left to right, the gluing of cut 2 to cut 1 compares each of 6 objects to each of $4(U+1)$ objects (local flows over cut 2, giving $24(U+1)$ operations), resulting in 6 local flows over cuts 1 and 2.
  Gluing cut 3 on requires another comparison of 6 objects to $4(U+1)$ objects, for a total of $48(U+1)$ operations.

  The brute force method, on the other hand, requires examining $6 \times 4(U+1) \times 4(U+1) = 96(U+1)^2$ coordinates and checking that inflow equals outflow at all nodes.
  This results in $384(U+1)^2$ operations, $8U$ times that required for our gluing algorithm.

  \begin{figure}[ht!]
    \centering
    \begin{tikzpicture}[
	roundnode/.style={circle, draw=green!30, fill=green!5, very thick, minimum size=7mm, scale=.7},
	roundnode2/.style={circle, draw=black!100, fill=green!0, thick, minimum size=7mm, scale=.7},
	squarednode/.style={rectangle, draw=red!60, fill=red!5, very thick, minimum size=5mm},
	el/.style = {inner sep=2pt, align=left, scale=.7},
      ]
      \node[roundnode2] at (0,0) (s) {$s$};
      \node[roundnode2] at (3,0) (v1) {$v_1$};
      \node[roundnode2] at (5,0) (v2) {$v_2$};
      \node[roundnode2] at (7,0) (t) {$t$};
      \path
      (s) edge [->] node[el,above] {$(0 \leq x+y \leq 2)$} (v1)
      (v1) edge [->] node[el,above] {$x\in [0,U]$\\$y\in [0,3]$ } (v2)
      (v2) edge [->] node[el, above] {$y\in [0,U]$\\$x\in [0,3]$} (t);

    \end{tikzpicture}
    \caption{A multicommodity network.
      Capacities are restricted to integer values of each commodity within the intervals listed on the edges.
      The parameter $U$ is a large number.
      }
    \label{ex:multi_com}
  \end{figure}
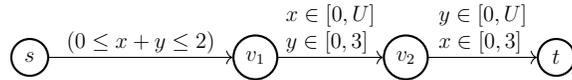
\end{example}

\section{Framework Based on the Work of Krishnan}\label{sec:Sheaf}
We briefly discuss a theorem presented by Krishnan \cite{Kr20??final} and discuss a way to frame the concept of mutual capacity using the structures through which he studied this problem.

Krishnan \cite{Kr20??final} provides two cases where the duality gap disappears and the maximum flow and minimum cut coincide.
First, if each $s$--$t$ path passes through a local minimum, i.e., an arc whose  capacity equals the intersection of all capacities for arcs along that path, we can avoid the duality gap.
Second, the duality gap vanishes if we can model the capacities as sufficiently nice objects in the sense of the following theorem. 

\begin{theorem}(\cite[Corollary~7.12]{Kr20??final}) \label{thm:Semi_MFMC}
Consider the following data.

(1) Lattice-ordered flat $\mathbb{N}$-semimodule $M$.

(2) Finite $M$-weighted digraph $(X;\omega)$ with distinguished edge $e$.

Then $\sup_{\text{flow  } \phi}\phi(e) = \inf_{\text{e-cut }C}\Sigma_{c\in C}\omega_c$.
\end{theorem}

We wish to build such a lattice-ordered flat $\mathbb{N}$-semimodule which will model our version of the multicommodity flow problem.
To this end, we first introduce definitions on and certain properties that such a set must satisfy.
All definitions in this section are taken from Krishnan's paper \cite{Kr20??final}.


\begin{definition}[Lattice]
  A lattice is a poset $S$ such that each pair of elements $s_1,s_2$ admits a least upper bound  $s_1 \smallwedge s_2$ and a greatest lower bound $s_1\smallvee s_2$.
  A lattice is complete if there exist a minimum $0_S$ and a maximum $1_S$.
\end{definition}

\begin{definition}[Partial monoid]
  A partial monoid is a set $\Mnd$ with operation $+_{\Mnd} : \Mnd \times \Mnd \rightarrow \Mnd$ and element $0_{\Mnd}$ such that in the following equations both sides are defined and coincide if one side is defined:
\begin{align*}
0_{\Mnd}+_{\Mnd} x & = x ~\mbox{ and } \\
x+_{\Mnd} (y+_{\Mnd} z) & = (x +_{\Mnd} y)+_{\Mnd} z.
\end{align*}
\end{definition}

\begin{definition}[$\mathbf{S}$-semimodule]
  A partial commutative monoid $\Mnd$ is an $S$-semimodule provided there exists a form of scalar multiplication $\times_{\Mnd} : S \times \Mnd \rightarrow \Mnd$ such that for the following equations both sides are defined and coincide if one side exists:
\begin{align*}
  (\lambda_1 \times_S \lambda_2)\times_{\Mnd} x & = \lambda_1 \times_{\Mnd} (\lambda_2\times_{\Mnd} x); \\
  (\lambda_1 +_S \lambda_2)\times_{\Mnd} x & = (\lambda_1\times_{\Mnd} x)+_{\Mnd} (\lambda_2\times_{\Mnd}  x); \\
  \lambda \times_{\Mnd}(a+_{\Mnd} b) & = (\lambda \times_{\Mnd} a)+_{\Mnd} (\lambda \times_{\Mnd} b); \\
  1_S\times_{\Mnd}  x & = x; \\
  0_S\times_{\Mnd} x & = 0_{\Mnd};~\mbox{ and } \\
  \lambda \times_{\Mnd} 0_{\Mnd} & = 0_{\Mnd}.
\end{align*}
\end{definition}

\begin{definition}[Partially Ordered Commutative Monoid]
  A partially ordered commutative monoid $\Mnd$ is a commutative monoid with a reflexive, transitive, antisymmetric partial order $\leq_{\Mnd}$ such that $a+_{\Mnd} b\leq_{\Mnd} a+_{\Mnd} c$ whenever $b\leq_{\Mnd} c$.
\end{definition}

\begin{definition}[lattice ordered commutative monoid]
A lattice ordered commutative monoid $\Mnd$ is a complete (i.e. the underlying poset is a complete lattice) partially ordered partial monoid such that $a+_{\Mnd} b\smallvee a+_{\Mnd} c=a+_{\Mnd}(b\smallvee c)$.
\end{definition}

\subsection{Construction} \label{ssec:sheafcnstn}

Now let $\Mnd$ be constructed as follows: $c_1, c_2, \dots, c_m$ be the capacities of the $m$ arcs in a network (these will be regions in $\mathbb{R}^k$ in the network with $k$ commodities), and let $\mathcal{C} = c_1 \times c_2 \times \dots \times c_m$.
Let $\{\pi_1, \pi_2, \dots, \pi_m\} \in \Mnd$, with $\pi_i=$ the canonical projection onto the $i$th component of $\mathcal{C}$.
This captures the capacity restrictions on each arc.
The image of $\mathcal{C}$ under $\pi_i$ gives the set of feasible flow values that can flow through arc $i$, and the level set for any value gives the set of flows which achieve that value on that arc.

For $f_1, f_2 \in \Mnd$, let $f_1+_{\Mnd} f_2=f_3$ be standard function addition, and let $(f_1\cdot_{\Mnd} f_2)(c)=f_1(c)$ if $f_1(c)=f_2(c)$, and $-\infty$ otherwise. Let $f_1 \leq_{\Mnd} f_2$ if and only if $f_1(c)=f_2(c)$ or $-\infty$ for all $c\in \mathcal{C}$, or if $f_2=\infty$ $\forall c\in \mathcal{C}$.

Addition is essentially the same idea as taking Minkowski sums of the images of the functions, except it requires the vectors to be associated with at least one of the same flows.
Multiplication, which will be the greatest lower bound in our lattice, captures the values and flows which satisfy both functions at the same time.
A function $f$ is less than or equal to another function provided they agree wherever $f$ is defined.
We also include an ``infinity'' function which is greater than or equal to all functions.

Let $g(c)=\infty$ for all $c$, and let $0_{\Mnd}=s$ where $s(c)=0$ for all $c\in \mathcal{C}$. 
We need a universal upper and lower bound, which we get in the form of $g$ and $-g$, respectively.
We also need the $0$ flow, which we get in the form of $0_{\Mnd}$.

Let $\Mnd^-$ be the set generated by $\{\pi_1, \pi_2, \dots, \pi_m\}$  under $+_{\Mnd},\times_{\Mnd}$, and function scalar multiplication by $-1$.
Let $\Mnd=\Mnd^-\cup g\cup 0_{\Mnd}$.

Now we will show that $\Mnd$ has all the required properties to satisfy Theorem \ref{thm:Semi_MFMC}.

\begin{lemma}
  The underlying poset of $\Mnd$ is a complete lattice.
\end{lemma}

\begin{proof}
  We claim that $f_1\smallwedge f_2=f_1$ if $f_2\leq_{\Mnd}  f_1$, vice-versa if $f_1\leq_{\Mnd} f_2$, and if neither of those holds $f_1\smallwedge f_2=g$.
  If one function is $\leq_{\Mnd}$ the other, e.g., suppose $f_1 \leq_{\Mnd} f_2$ without loss of generality, then by  definition any function greater than both must be greater than $f_2$, and hence $f_2$ is the lowest upper bound on $f_1,f_2$.
  Otherwise, there exists a $c\in \mathcal{C}$ where $-\infty\neq f_1(c)\neq f_2(c)\neq -\infty$.
  Then there is no function which agrees with both on $c$, and hence only $g\geq_{\Mnd} f_1,f_2$.
\end{proof}

\begin{lemma}
  $f_1\smallvee f_2=f_1\cdot_{\Mnd} f_2$.
\end{lemma}

\begin{proof}
  Suppose $f_3\leq_{\Mnd} f_1$ and $f_3\leq_{\Mnd} f_2$.
  Then for all  $c\in \mathcal{C}$ $f_3$ agrees with $f_1$ or maps to $-\infty$ and $f_3$ agrees with $f_2$ or maps to $-\infty$.
  Then $f_3$ maps all values on which $f_1$ and $f_2$ don't agree to $-\infty$, so $f_3\leq_{\Mnd} f_1\cdot_{\Mnd} f_2$.
\end{proof}

\begin{lemma}
  $\Mnd$ is a partial monoid.
\end{lemma}

\begin{proof}
  We observe that $0_{\Mnd}+_{\Mnd} f=f$ for all $f$, and function addition is associative.
\end{proof}

\begin{lemma}
  $\Mnd$ is a commutative $\mathbb{N}$-semimodule.
\end{lemma}

\begin{proof}
  Since the functions over real numbers are a semimodule over $\mathbb{N}$, $\Mnd$ inherits the required properties to be a commutative $\mathbb{N}$-semimodule.\end{proof}

\begin{lemma}
  $\Mnd$ is a partially ordered commutative monoid.
\end{lemma}

\begin{proof}
  Clearly $\leq_{\Mnd}$ is reflexive, transitive, and antisymmetric.
  Furthermore, suppose $f_1\leq_{\Mnd} f_2$.
  Then for all $c\in \mathcal{C}$, $f_1(c)=f_2(c)$ or $f_1(c)=-\infty$.
  Then $(f_1+_{\Mnd} f_3)(c)=f_1(c)+f_3(c)$ whenever $f_1(c)=f_2(c)$, and $(f_1+f_3)(c)=-\infty+f_3(c)=-\infty$ otherwise.
  Thus $\Mnd$ is a partially ordered commutative monoid.
\end{proof}

\begin{lemma}
  $\Mnd$ is a lattice ordered commutative monoid.
\end{lemma}

\begin{proof}
  Since $(f_1+_{\Mnd} f_2)\smallvee(f_1+_{\Mnd} f_3)=(f_1+_{\Mnd} f_2)(c)$ when $(f_1+_{\Mnd} f_2)(c)=(f_1+_{\Mnd} f_3)(c)$ and $-\infty$ otherwise,
  and since $(f_1+_{\Mnd} f_2)(c)=(f_1+_{\Mnd} f_3)(c)$ if and only if $f_2(c)=f_3(c)$ or $f_1(c)=-\infty$,
  then for $c\in \mathcal{C}$ either $f_1(c)=-\infty$, and hence $(f_1+(f_2\times_{\Mnd} f_3))(c)=-\infty= ((f_1+_{\Mnd} f_2)\times_{\Mnd} (f_1+_{\Mnd} f_3))(c)$;
  or $f_2(c)=f_3(c)$, and thus $((f_1+_{\Mnd} f_2)\times_{\Mnd} (f_1+_{\Mnd} f_3))(c)=f_1(c)+(f_2\times_{\Mnd} f_3)(c)$.
  Thus $\Mnd$ is a lattice ordered commutative monoid.
\end{proof}

Note that $\mathbb{N}$ is flat over itself and direct sums of flat modules are flat \cite{nlab:flat_module}.
Hence we can build flat $\mathbb{N}$ semimodules representing the flows over any collection of edges, and subsemimodules of any such construction remain flat \cite{Kr20??final}.

Thus Corollary 7.12 from Krishnan \cite{Kr20??final} gives us that, provided the semimodule for which $\Mnd$ is a subsemimodule is flat over $\mathbb{N}$ (such as, for instance, any set of capacities modeling rational flows), $\sup_{\text{flow } \phi} \phi(e) =\inf_{\text{e-cut }C}\Sigma_{c\in C}\omega_c$.
In other words, the feasible region of flows is given by the infimum over e-cuts of the sum of the projection functions associated with that cut.
The region of feasible flow values must then be the image of $\mathcal{C}$ under that function, while all possible flows achieving that value is determined by the level sets of the function.

\section{Ratio Maximum Flow} \label{sec:ratiomaxflow}
While Theorem \ref{thm:MCMFMC} presents a tight duality result for a general multicommodity max flow problem, the resulting algorithms presented in Section \ref{ssec:MCAlgo} could still be computationally expensive.
Hence we consider problems closely related to the \MCMF problem for which we could develop computationally tractable approaches.
Motivated by the work of Leighton and Rao \cite{LeRa1999}, we study \RMF and \IRMF as multicommodity maximum flow problems that include an explicit optimization component rather than searching for the full feasible region.
Our approach is motivated by augmenting path algorithms for single commodity max flow problem \cite{AhMaOr1993}.
However, we search for an augmenting \emph{cycle set} instead of augmenting paths.

\begin{definition}[Augmenting Cycle Set]
  An \emph{augmenting cycle set} in a network with flow $f$ having coefficients in $\mathbb{R}^k$ is a set of cycles such that if we modify the flow of every arc in a cycle by that cycle's coefficients,
  we obtain a new flow with a higher value than $f$.
\end{definition}

\begin{lemma} \label{lem:flocycle}
  Any two flows with the same flow value in a network differ by an element of the cycle space.
\end{lemma}

\begin{proof}
  We first observe that this result is elementary in the case of single commodity flows.
  we show that it also holds for our general multicommodity setting.
  Let flows $f_1$ and $f_2$ in the enhanced network have the same flow value.
  Since both flows are cycles in the enhanced network, $f_1-f_2$ must also be in the cycle space.
  Furthermore, since both flows have the same flow value, and hence the same value on the distinguished edge $e$, $f_1-f_2$ has a value $0$ on edge $e$, and thus lies in the cycle space of the base network.
  Thus it follows that any two flows with the same flow value differ by a sum of cycles.
\end{proof}

We first address the following decision problem:
given a flow value $f$, is there a flow on the network $N$ with value $f$?
We solve this problem by constructing a pseudoflow $F'$ with value $f$ and search for a set of augmenting cycles that transform $F'$ into a flow.
To this end, we transform the problem into one of finding an area of cycle space within particular bounds.
We show that the decision problem returns true if and only if the bounds in cycle space are non-empty, and that augmenting by a set of cycles in that region gives a flow.

Since we know that the flows on a network are cycles in the enhanced network, and since we know two potential flows share a flow value if and only if they differ by a cycle in the network (Lemma \ref{lem:flocycle}), we can approach the problem of finding a valid realization of a particular flow value by considering it in terms of the cycle space of the network.
Since all arcs in our network are intrinsically bi-directional, the cycle space of our network is in fact the cycle space of the underlying graph.
A generating set for a graph's cycle space can be found in polynomial time by first finding a spanning tree of the graph rooted at $s$ \cite{AhMaOr1993}.
Each edge not on the spanning tree determines a unique cycle with the tree:
take the unique path to the head and tail of the edge and include in our cycle each edge unique to one of the paths, then add the edge in question not on the tree.
There are then $n-1$ edges on the tree and $m-n+1$ edges determining cycles. 

Next we establish constraints on what combinations of cycles will give a flow.
We start by partitioning the edge set by the cycles to which the edge belongs.
The $m-n+1$ edges determining unique cycles belong to exactly one cycle.
The remaining $n-1$ edges may belong to one or more cycles, and we can find this information directly from the cycle-membership matrix.
Finally, we determine the capacity constraints for each set by intersecting the orientation corrected capacity of each arc involved in that set.

\begin{definition}[Cycle Constraints] \label{def:cycons}
  For a network with a set of generating cycles $\{C_i\}_{i\in I}$, the \emph{cycle constraints} $\{c_i\}_{i\in I}$ are specified on $k$-dimensional vectors.
  For each arc $a$ in the network, we have a constraint $R_a$ that the sum of all $c_i$ that use $a$ as a forward arc minus the sum of all $c_i$ that use $a$ as a backward arc must lie within the capacity of $a$ minus the flow assigned to $a$.
\end{definition}

We illustrate cycle constraints on an example network in Figure \ref{fig:cycle_search}.
Suppose we have cycles $c_1=\{a_1,a_3,-a_4,-a_2\}$, $c_2=\{a_5,a_8,-a_3\}$, and $c_3=\{a_5,a_7,-a_6,-a_3\}$ in the network.

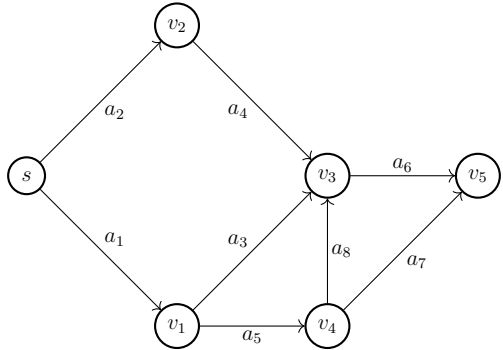
\begin{figure}[ht!]
    \centering
\begin{tikzpicture}[
roundnode/.style={circle, draw=green!30, fill=green!5, very thick, minimum size=7mm, scale=.7},
roundnode2/.style={circle, draw=black!100, fill=green!0, thick, minimum size=7mm, scale=.7},
squarednode/.style={rectangle, draw=red!60, fill=red!5, very thick, minimum size=5mm},
 el/.style = {inner sep=2pt, align=left, scale=.7},
]
\node[roundnode2] (s) at (0,0) { $s$};
\node[roundnode2] (v1) at (2,-2) { $v_1$};
\node[roundnode2] (v2) at (2,2) { $v_2$};
\node[roundnode2] (v3) at (4,0) { $v_3$};
\node[roundnode2] (v4) at (4,-2) { $v_4$};
\node[roundnode2] (v5) at (6,0) { $v_5$};

\path
 (s) edge [->] node[el,above right] { $a_1$} (v1)
 (s) edge [->] node[el, below right ] { $a_2$} (v2)
 (v1) edge [->] node[el, above left] { $a_3$} (v3)
 (v2) edge [->] node[el, below left] { $a_4$} (v3)
 (v3) edge [->] node[el, above] { $a_6$} (v5)
 (v4) edge [->] node[el, below right] { $a_7$} (v5)
 (v4) edge [->] node[el, right] { $a_8$} (v3)
 (v1) edge [->] node[el, below] { $a_5$} (v4);
 
 \end{tikzpicture}
    \caption{Example network to illustrate cycle constraints in Definition \ref{def:cycons}.
    }
    \label{fig:cycle_search}
 \end{figure}
 
 Then to find the cycle constraints, we observe $a_1, a_2, a_4$ are unique to $c_1$, $a_3$ is common to all three, $a_5$ is unique to $c_2$ and $c_3$, $a_6, a_7$  are unique to $c_3$, and $a_8$ is unique to $c_2$.
 Then $c_1$ has a capacity of $v(a_1) \cap -v(a_4) \cap -v(a_2)$, $c_2$ has a capacity of $v(a_8)$, and $c_3$ has a capacity of $-v(a_6) \cap v(a_7)$.
 Additionally, $c_2+c_3$ has capacity $v(a_5)$, and $c_1-c_2-c_3$ has capacity $v(a_3)$.
 
Then any feasible solution to $c_1,c_2,c_3$ such that:
\begin{itemize}[itemsep=-0.02in]

  \item $c_1 \in$ $v(a_1) \cap -v(a_4) \cap -v(a_2)$

  \item $c_2 \in$ $v(a_8)$

  \item $c_3 \in$ $-v(a_6) \cap v(a_7)$

  \item $c_2+c_3 \in$ $v(a_5)$

  \item $c_1-c_2-c_3 \in$ $v(a_3)$
\end{itemize}
is a set of cycle augmentations that give a feasible flow of the current infeasible flow value.

In general, each arc contributes a constraint that the sum of all generating cycles through that arc in the positive orientation minus the sum of generating cycles through the arc in the negative orientation must lie within the adjusted capacity of that arc.
Hence we get an $(m-n+1)k$-dimensional optimization problem with $m$ constraints.

\begin{theorem}
  Augmenting along a cycle set results in a flow if and only if the cycle set abides by the cycle constraints of the network.
\end{theorem}

\begin{proof}
  Suppose we have a cycle set in the solution set to the problem.
  Then the constraints imposed by arcs unique to a single cycle $C$ restrict the possible values of $C$ to ones that result in no arc unique to that cycle overflowing when augmented along.
  Likewise, the constraints imposed by the arcs common to two cycles restrict the values those cycles can take to pairs of values which will result in none of the common arcs overflowing after augmentation, and so on.
  Thus since all arcs belong to one set in the partition, and since all constraints are satisfied, augmenting along the cycle set must result in no arcs overflowing.
  Since we started with a pseudoflow, which respects flow conservation, and augmented along cycles (which preserve flow conservation) we get a flow.
  Since augmenting along cycles does not change the value of the flow, we maintain a flow value of $f$.

  Suppose a set of cycles gives us a flow after augmentation.
  Then for each arc, the cycle set must fall within the constraints imposed by that arc.
  Thus the cycle set must fall in the solution set to the above problem.
  Again, augmenting along cycles guarantees that we maintain the same flow value.
\end{proof}

We use the solution to the decision problem in algorithms that find an $\epsilon$-approximate solution to \RMF (Algorithm \ref{alg:CSearch}) and an $\epsilon$-approximate solution to \IRMF using a binary search (Algorithm \ref{alg:IntCSearch}).

\begin{algorithm}[ht!]
 \KwData{Network $N$ with compact (i.e., closed and bounded)
 reducible capacities $\{C_i\}$;\\  \hspace*{0.4in} 
 Desired Ratio: $R$, the ratio of commodities we wish to transport;\\  \hspace*{0.4in} 
 Upper bound: $B^+$, an overestimate on the multiples of $R$ that can be sent;\\  \hspace*{0.4in} 
 Lower bound: $B^-=0$; and\\  \hspace*{0.4in} 
 Error allowed: $\epsilon>0$ (returned multiple is within  $\epsilon$ of true max flow multiple).}
 \KwResult{$P^*$, an approximation of the true max flow multiple $P$ such that $P^*R \leq PR \leq (P^*+\epsilon)R$}

 \SetKwFunction{FTest}{Test}
 
 \vspace*{0.1in}
 Initialization:
 
 $D=B^+-B^-$\;
 
 \While{$D>\epsilon$}{
 $T=\frac{B^++B^-}{2}$\;
 Test$(TR,B^-,B^+)$\;
 $D=B^+-B^-$\;
 }
 \KwRet{$B^-$}

 \vspace*{0.1in}
 \SetKwProg{Fn}{Function}{:}{}
 \Fn{\FTest{$V$, $B^-$, $B^+$}}{
   $T=\frac{B^++B^-}{2}$\;
    Test decision problem on $V$\;
   \eIf{decision problem returns TRUE}{
    $B^+=T$}
     {
     $B^-=T$  
     }}

 \vspace{0.5\baselineskip}
 \caption{$\epsilon$-error approximation for \RMFs.}
 \label{alg:CSearch}
\end{algorithm}

\begin{algorithm}[ht!]
 \KwData{Same as in Algorithm \ref{alg:CSearch}, except we require $0<\epsilon<0.5$ \\ \hspace*{0.4in}
 (so that the returned value can be within $\epsilon$ of at most one integer).}
 \KwResult{$P^*$, an approximation of the true max flow multiple $P$ such that $P^*R \leq PR \leq (P^*+\epsilon)R$}
 
 \SetKwFunction{FTest}{Test}
 \SetKwFunction{FEndTest}{EndTest}

 Initialization:
 
 $B^-=0$\;
 
 $D=B^+-B^-$\;
 
 \While{$D>\epsilon$}{
 $T=\frac{B^++B^-}{2}$\;
 Test$(TR,B^-,B^+)$\;
 $D=B^+-B^-$\;
 }
 
 EndTest$(\ceil{B^-}R,B^-)$

 \KwRet{$B^-$}

 \vspace*{0.1in}
 \SetKwProg{Fn}{Function}{:}{}
 \Fn{\FEndTest{$V$, $B^-$}}{
   Test decision problem on $V$\;
   \eIf{decision problem returns TRUE}{
    $B^-=\ceil{B^-}$}
     {
     $B^-=\floor{B^-}$
     }}

 \vspace{0.5\baselineskip}
 \caption{$\epsilon$-error approximation for \IRMFs.
 The function \texttt{Test} is defined in Algorithm \ref{alg:CSearch}.}
 \label{alg:IntCSearch}
\end{algorithm}

We show that these algorithms terminate and that they are accurate.

\begin{lemma} \label{lem:AlgosTerm}
  Algorithms \ref{alg:CSearch} and \ref{alg:IntCSearch} terminate.
\end{lemma}

\begin{proof}
  Because the capacities of the network are bounded, we get a finite initial value for $B^+$.
  As a binary search, the algorithm reduces the difference of $B^+$ and $B^-$ by a factor of $\frac{1}{2}$ at each step.
  Thus $B^+-B^-<\epsilon$ when $\frac{B^+}{2^s}<\epsilon$ after $s$ steps, which happens in $s = \ceil{\log(\frac{B^+}{\epsilon})}$ steps.
\end{proof}

\begin{theorem} \label{thm:epsapprox}
  Algorithm \ref{alg:CSearch} gives a feasible max flow in the desired ratio which is within $\epsilon$ of the true max flow.
\end{theorem}

\begin{proof}
  Since all capacities are reducible, the $\mathbf{0}$ flow is feasible.
  Since $B^-$ initializes at $0$ and is replaced only with feasible multiples, $B^-$ is guaranteed to be feasible at all steps.

  Now suppose there exists $P'\in \mathbb{R}^+$ such that $P'R$ has a realization, but $P'-B^->\epsilon$.
  We have $B^+-B^-<\epsilon$ by the termination condition, so $P'>B^+$.

  Since $P'$ has a realization, multiplying the values at each arc by $\frac{B^+}{P'}$ gives a realization for $B^+$, which is feasible because each arc has a reducible capacity.
  Thus $B^+$ has a realization.
  But the algorithm only assigns values for which no realization exists to $B^+$, so this is a contradiction.

  Hence the max flow is within $\epsilon$ of $B^-$.
\end{proof}

\begin{lemma} \label{lem:alg:IntCSearch}
  Algorithm \ref{alg:IntCSearch} gives a max flow in the desired ratio which the greatest integer multiple of the desired ratio which can pass through the network.
\end{lemma}

\begin{proof}
  By Theorem \ref{thm:epsapprox}, $B^-$ is within $\epsilon$ of the true max flow.
  Since $\epsilon < 0.5$ here, $B^-$ can be within $\epsilon$ of at most a single integer.
  Since $B^-$ has a realization, the maximum integer value with a realization must be either $\ceil{B^-}$ or $\floor{B^-}$, and can be $\ceil{B^-}$ only if $\ceil{B^-}-B^-<\epsilon$.
  Since Algorithm \ref{alg:IntCSearch} checks $\ceil{B^-}$ in such a case, it must return the maximum integer value with a feasible realization.
\end{proof}

\section{Discussion} \label{Discussion}
We addressed two problems using topological techniques and with graphs to gain insight into data and complex systems.
First, we discussed our work on the \MCMF problem, a generalization of the traditional Max Flow problem from the study of network flows.
The generalization removes some key pieces from the traditional Max-Flow Min-Cut theorem and the algorithms to which it gives rise---in particular, the idea of a flow saturating a cut does not translate nicely into the multicommodity setting.
Our work developed a way to connect the feasible region for a network flow problem to the flows that can fit through all $s$--$t$ cuts in the network.
We have also given some methods to calculate those values for the general problem and some useful sub-problems.
It would be interesting to extend the results in Section \ref{ssec:DisjointMutualCap} on fully disjoint networks to the case of multicommodity flows in \emph{series-parallel digraphs}, which are obtained by composing in an edge-disjoint manner (i.e., by merging the source and sink nodes) the copies of smaller digraphs in series or in parallel fashion. 

Our main result is Theorem \ref{thm:MCMFMC}, in which we showed that the set of local flows over individual cuts that can be ``glued together'' into true flows is precisely the set of feasible flows over the network.
By generalizing max flow to be the feasible region of flow through the network and the min cut to be the set of consistent collections of local flows over cuts, we achieve a generalization of the duality between flows and cuts which the traditional Max-Flow Min-Cut theorem provides.

We also studied a special case of \MCMF that considers maximizing the multiple of a specific ratio of flows transported through the network.
We show how to model the possible rearrangements of flow as a problem on the cycle space of the network.
We show that this method of assigning coefficients exactly captures the set of feasible rearrangements of flows of particular flow values, meaning that a pseudoflow has a feasible realization if and only if it has a feasible set of cycle coefficients.
This gives us a method to check the feasibility of a particular flow value which, in turn, allows us to use a binary search of possible multiples of the desired ratio to determine an arbitrarily precise approximation of the max flow. 

There are several avenues for future research in this area.
Primarily, more efficient methods for finding the mutual capacities for a collection of cuts is necessary for any large-scale implementations of this work.
Even efficient methods for calculating better approximations, e.g., mutual capacities for all collections of three cuts, could give useful heuristics for solving max-flow problems.

Industry applications of multicommodity flows may be solvable by more computationally tractable subproblems.
Finding additional assumptions about network structure (e.g., acyclic networks or limits on degrees of nodes) or capacities (e.g., polygonal regions or convexity) that give more computationally efficient solutions is another potentially fruitful area of research.

\vspace*{-0.12in}
\paragraph{Acknowledgment:}  \label{acknowledgment}
We thank the National Science Foundation for support through grants 1661348 and 1819229.

\end{document}